\documentclass[num-refs]{wiley-article}
\usepackage{amsmath}
\usepackage{amssymb}
\usepackage{anyfontsize}
\usepackage{graphicx}
\usepackage{subcaption}
\usepackage{float}
\usepackage[english]{babel}
\usepackage{tikz,tkz-tab}
\usepackage{textcomp}
\usepackage{datetime}
\usepackage{mathtools}
\usepackage{xcolor}
\usepackage{color}
\usepackage[bookmarks=false]{hyperref}
\usepackage{algorithm}
\usepackage{algpseudocode}
\usepackage{multirow}

\usepackage[font=small,labelfont=bf]{caption}
\usepackage{natbib}
\setcitestyle{numbers,open={[},close={]}}

\newtheorem{remark}[theorem]{Remark}

\newcommand{\EE}{{\mathbb E}}

\newcommand{\PP}{{\mathbb P}}

\def \Fc{{\cal F}}

\def \g{\gamma}

\def \k{\kappa}
\def \l{\lambda}

\def \n{\nu}

\def \r{\rho}

\def \f{\varphi}
\def \o{\omega}

\def \e{\varepsilon}

\renewcommand{\thefootnote}{\fnsymbol{footnote}}

\title{Efficient approximations for utility-based pricing}

\author[2]{Laurence Carassus}
\author[1,2,\thanks{}]{Massinissa Ferhoune}

\affil[1]{L\'{e}onard de Vinci P\^ole Universitaire, Research Center, 92 916 Paris La D\'{e}fense, France}
\affil[2]{Laboratoire de Math\'{e}matiques de Reims, UMR9008 CNRS et Universit\'{e} de Reims Champagne-Ardenne, France}

\corraddress{Massinissa Ferhoune}
\corremail{fmassinissa@free.fr}

\runningauthor{Laurence Carassus and Massinissa Ferhoune}

\begin{document}

\maketitle
\begin{abstract}
In a context of illiquidity, the reservation price  is a well-accepted alternative to the usual martingale approach which does not apply. However, this price is not available in closed form and requires numerical methods such as Monte Carlo or polynomial approximations to evaluate it. We show that these methods can be inaccurate and propose a deterministic decomposition of the reservation price using the Lambert function. This decomposition allows us to perform an improved Monte Carlo method, which we name Lambert Monte Carlo (LMC) and to give deterministic approximations of the reservation price and of the optimal strategies based on the Lambert function. We also give an answer to the problem of selecting a hedging asset that minimizes the reservation price and also the cash invested.
 Our theoretical results are illustrated by numerical simulations.\\
\vspace{6pt}
\end{abstract}


\section{Introduction}
\let\thefootnote\relax\footnotetext{*Corresponding author, email : fmassinissa@free.fr}We consider a continuous-time financial model with time horizon $T$. An economic agent will receive, at time $T,$ $\lambda$ units of a derivative written on some non-traded stock $S$ and seek to evaluate the price she is willing to pay for that.  
This is the classical problem of pricing in mathematical finance, but as the investor is facing trading restriction, she wants to estimate the risk premium related to this lack of marketability. Indeed, in a complete market when $S$ can be dynamically traded (i.e. can be bought and sold at any time between time $0$ and time $T$), the derivative price is given by its discounted expected value under some martingale measure and the marketability risk premium is null. 
As this approach is not feasible when $S$ is not dynamically traded, we assume that the investor hedges her position on the derivative on $S$ by trading another \lq\lq{}similar\rq{}\rq{} liquid asset $P.$ The asset $P$ must be thought of as a related index or as another stock from a closed industry sector. We follow \cite{refHN} and use the reservation price concept in order to value the illiquidity of $S$. 
The reservation price (also known as indifference price) is the price $p$ which leaves the agent indifferent between paying $p$ at time $0$ and receiving $\l$ units of shares of the derivative on $S$ at time $T,$ or doing nothing. This is a preference-based bid price. In the case of a Black-Scholes-like market, when the preferences of the agent are represented by some exponential utility function with risk aversion coefficient $\gamma$, this problem has been well-studied in the literature: see \cite{ref2, ref1,ref3, H94, refHH} and \cite{HM15} and the references therein.
It is possible to obtain a semi-closed formula for the reservation price which involves the logarithm of the Laplace transform of some function of a lognormal random variable. The expectation in the Laplace transform may be computed with the Monte Carlo method. 
  Nevertheless, we observe that the Monte Carlo method leads to unreliable estimators suffering from very high variance (see Figures \ref{lambert_direct_stock} and \ref{lambert_direct_put}). Thus, our first motivation is to find an improved numerical method for the computation of the reservation price. 

An alternative method to Monte Carlo is to produce a polynomial approximation. This has been done for an expansion in terms of small $\l$ in \cite{ref1,ref3}, of the correlation $\rho$ between $P$ and $S$ in \cite{ref2} and of an aggregate of  $\gamma$, $\l$ and $\rho$ in \cite{ref4} and \cite{ref5}. Thus, we first derive polynomial approximations in $\rho$ and show (see Figure \ref{figure_dl}) that these polynomial approximations can be extremely bad. We take the opportunity to precise the error order of the expansions proposed by Monoyios in \cite{ref4} and \cite{ref5}, as its coefficients are defined using the minimal martingale measure and thus depend on the correlation. 

Based on these observations, we propose a new and improved numerical method. For this method, we adapt an idea from \cite{ref6} and \cite{ref12}, where a modified Laplace method is introduced in order to get a closed form approximation of the Laplace transform of a lognormal random variable. 
The Lambert function, which is the inverse of $x \mapsto x e^x,$ appears in this approximation. This function has been widely used and study in the last 30 years due to the advent of fast computational methods, see \cite{ref13}. Here, we adapt the idea of \citep{ref6} considering a non smooth function of a lognormal random variable.
This modified Laplace method allows splitting the reservation price into a deterministic part and a random part, the first one being preponderant in meaningful financial situations, while the second, which is close to one, will help to provide a low variance estimator for the price. 
This method also provides some multiplicative decomposition for the value function. We call this method the Lambert Monte Carlo (LMC) method. 
We show numerically that the LMC method is indeed an improvement of the Direct Monte Carlo (DMC) method in various market situations (see Figures \ref{lambert_direct_stock} and  \ref{lambert_direct_put}). We also compare the speed of both methods and find numerically that they are of same order.
But the number of simulations required for the DMC method to obtain a satisfactory confidence interval can be very large. This is not the case for the LMC method (see Table \ref{number_n}). We refer to \citep{ref6} for numerical comparison of the LMC method with methods other than polynomial approximation available in the literature as well as a survey on the numerical methods developed to approximate the Laplace transform. Here, we will focus on comparison with series expansion as this is the standard for the computation of the reservation price in the financial literature.

We then apply the preceding results to the case of an agent willing to cover a long position on the non-traded stock. This is a very important issue in management sciences which is often referred to as the restricted stock problem or marketability problem: see \cite{L95, KLF03, F12}, or for the real option problem, see \cite{ref3} and the references therein. 
 The second main result of this paper is to provide a lower and an upper deterministic approximation of the reservation price, the value function as well as an approximation of the optimal strategy and of some Greeks. They are based on our decomposition. 
We study the efficiency of our approximations in terms of $\lambda$, $\rho,$ $\gamma$ and $s_0$ and show that they are indeed reliable in various financial situations (see for example Figure \ref{lambert_direct_stock}). They thus provide proxy which need almost no computation time. 
 We also compute the argminimum $\rho^*$ of our lower approximation of the reservation price as a function of $\rho$. We show numerically that $\rho^*$ is in fact a good approximation of the argminimum of $p$.  This minimum correlation offers a new perspective for utility-based pricing and hedging. Indeed, $\rho^*$ may be the optimal choice
for an agent selecting between different hedging assets $P$ who is willing to choose the one that minimizes the reservation price while retaining the same level of rentability and for which almost no money is invested in the (optimal) hedging strategy. We then show that choosing an hedging asset of correlation $\rho^*$ and using the (optimal) hedging strategy does not cripple the superhedging probability. We also apply our result to the case of a short position on a put option and provide numerical results.


The paper is organized as follows. In Section \ref{model}, we present the financial model. In Section \ref{taylorapprox}, we derive a polynomial approximation in the case of a long stock position. Section \ref{secdecomp} provides the decomposition of the reservation price and of the value function using the Lambert function. Then, Section \ref{seclongstock} analyses the case of a long stock position, while Section \ref{secshorput} examines the short put case. Section \ref{conclu} concludes and provides some overview on the applicability of our method to other financial problems. Finally, Section \ref{secappendix} regroups the missing proofs of the paper. 

\section{The financial model}
\label{model}
We consider a continuous time financial model with time horizon $T$ and a filtered probability space $(\Omega,\Fc,(\Fc_t)_{0\leq t \leq T},\mathbb{P}),$ where the filtration $(\Fc_t)_{{0\leq t \leq T}}$ satisfies the usual conditions (right continuous, $\Fc_0$ contains all null sets of $\Fc_{T}=\Fc$). 
We denote by $S=(S_t)_{0 \leq t \leq T}$ (resp. $P=(P_t)_{0 \leq t \leq T}$) the observable price process of some non-traded (resp. dynamically traded) asset. We assume that
\begin{eqnarray}
dS_t=S_t \left( \nu dt + \eta dZ_t \right)\quad dP_t=P_t \left( \mu dt + \sigma dB_t \right), \label{mdl}
\end{eqnarray}
where $Z=(Z_t)_{0 \leq t \leq T}$ and $B=(B_t)_{0 \leq t \leq T}$ are correlated standard Brownian motions under the historical probability $\PP$. The correlation is constant and is denoted by $\r\in (-1,1)$ and, as usual, we will write that $dZ_t=\rho dB_t + \sqrt{1-\rho^2} dW_t$, where $W$ and $B$ are independent Brownian motions under $\mathbb{P}$. 
The constants $\nu$ and $\mu$ (resp. $\eta$ and $\sigma$) respectively represent the expected return (resp. the volatility) of the non-traded and of the traded assets. 
There is also a riskless bond, with dynamic given by $D_t=e^{rt},$ where $r$ is the constant riskless interest rate. We introduce the Sharpe ratio of the traded asset $$s_R:= \frac{\mu-r}{\sigma}$$ as the measure of excess return per unit of risk. All discounted values for any process $ Y =({Y}_t)_{0 \leq t \leq T}$ will be denoted by $\widetilde Y =(\widetilde{Y}_t)_{0 \leq t \leq T},$ where $\widetilde{Y}_t={Y_t}/{D_t}$. 

The agent's initial wealth is denoted by $x$ and her trading strategy is denoted by $\Pi=(\Pi_t)_{0 \leq t \leq T}$, where $\Pi$ is a progressively measurable process that denotes the cash invested in the traded-asset $P$ at time $t$. We assume that $\int_{0}^{t}{\Pi_s^2 ds}<+\infty$ a.s. for all $t\geq 0.$ The agent wealth, starting from an initial wealth $x$ and following a self-financed strategy $\Pi$ is denoted by $X^{x,\Pi}=(X_t^{x,\Pi})_{0 \leq t \leq T}$ and as the agent cannot invest in $S$, her discounted wealth 
follows the dynamics: $d\widetilde{X}^{x,\Pi}_t=\Pi_t  \frac{d\widetilde{P}_t}{P_t}.$\\
To value the amount the agent will accept to pay in order to receive the derivative at time $T$, we use a utility-based price. The risk preferences of the agent are modeled by an exponential utility function, which is very popular because of its separability properties. Indeed, in contrast to the power utility function, it gives an explicit representation for the reservation price. Let
\begin{eqnarray*}
U(x)=-\frac{1}{\gamma}\exp{[-\g x]},
\end{eqnarray*}
where $\g >0$ is the constant absolute risk aversion parameter of the agent.\\
Let $V(x_0,\l,h)$ be the value function at time $0$ for an agent that maximizes her expected utility of terminal wealth at time $T,$ if in addition she receives $\l>0$ units of the European derivative $H=h(S_T)$, where $h$ is a measurable function whose assumptions will be specified below :
\begin{eqnarray}
\label{valeur}
V(x_0,\l,h)=\sup_{\Pi } \EE U\left(X_T^{x_0,\Pi} + \l h(S_T)\right).
\label{def_value}
\end{eqnarray}
It is possible to find a semi-closed formula for the optimal trading strategy and for the value function using either the dual approach based on martingale techniques (see \cite{REK} or \cite{KS} and the references therein) or the primal approach.  
The primal approach (see \cite{ref1} or \cite{ref3}) leads to a non-linear Hamilton-Jacobi-Bellman equation that can be linearized using the Hopf-Cole transformation.
Please refer to \cite{ref8}, where this transformation is called the distortion method. 
With a slight adjustment to Section 5 of \cite{ref1} (see also \citep[Remark 5]{ref14}), we determine that :
\begin{eqnarray}
\label{valeur2_gen}
V(x_0,\l,h)=-\frac{1}{\gamma}e^{-\gamma x_0 e^{rT}-\frac{(\mu-r)^2}{2\sigma^2}T}\left(\mathbb{E}\left(\exp\left[-\lambda\gamma(1-\rho^2)h\left(s_0 e^{\left(\nu - \eta \rho \frac{\mu-r}{\sigma}-\frac{\eta^2}{2}\right)T}e^{\eta\sqrt{T}N}\right)\right]\right)\right)^{\frac{1}{1-\rho^2}},
\end{eqnarray}
where $N$ is a standard Gaussian law. Note that the parameters of $P$ appear only through the Sharpe ratio $s_R$ of $P$. 
If $h$ is bounded from below on $(0,\infty)$, the bid reservation price $p_{h}$ of $\l>0$ units of the derivative $H=h(S_T)$ is the amount, which leaves the agent indifferent between paying $p_{h}$ at time $0$ and receiving $\l$ units of $H$ at time $T,$ or doing nothing, i.e., $p_h$ is a solution of the following equation:
\begin{eqnarray}
\label{indiff_def_gen}
V(x_0-p_{h},\l,h)=V(x_0,0,h).
\end{eqnarray}
If $h$ is bounded from above on $(0,\infty)$, the ask reservation price of $\l>0$ units of the derivative $H=h(S_T)$ is the amount, which leaves the agent indifferent between receiving $p^{sell}_{h}$ at time $0$ and delivering $\l$ units of $H$ at time $T,$ or doing nothing :
\begin{eqnarray}
V\left(x_0+p^{sell}_{h},\l,-h\right)=V\left(x_0,0,-h\right).
\end{eqnarray}
Thus, if $h$ is bounded on $(0,\infty)$, we easily determine that :
\begin{eqnarray}
\label{indiffsell_gen}
p_{h}^{sell}&=&-p_{-h}.
\end{eqnarray}
Then, (\ref{valeur2_gen}) and  (\ref{indiff_def_gen}) lead to 
\begin{eqnarray}
p_{h} & = & -\frac{e^{-rT}}{\g(1-\r^2)} \ln \EE \exp{\left[ -\l\g(1-\r^2) h\left(s_0e^{\left(\nu-\eta \rho s_R -\frac{\eta^2}{2}\right) T}e^{\eta\sqrt{T}N}\right)\right]} \label{indiffdef_gen}\\
V(x_0,\lambda,h)&=&-\frac{1}{\gamma}\exp\left[-\gamma e^{rT}(x_0+p_{h})-\frac{s_R^2}{2}T\right].\label{value_pe}
\end{eqnarray}
Thus, we will first search for approximations for $p_{h}$ and then deduce the results on $V$ using \eqref{value_pe}. 

Let $K\in [0,+\infty]$ and $\zeta$  be a measurable function,  which is bounded from below on $[0,K].$
In the following, we will focus on the computation of the bid reservation price of a derivative of the form $H=\zeta(S_T) 1_{S_T \leq K}.$   This includes the bid price of the non-traded stock $H=S_T$ choosing $K=\infty$ and $\zeta=id,$ where $id(x)=x$ for all $x\in \mathbb{R}.$ 
We will also get the ask price of a put option  $H=(K-S_T)_+$ (where $x_+=\max(x,0)$) for some $K\in [0,+\infty)$, choosing $\zeta(x)=x-K.$ Indeed, $p_{(K-x)_+}^{sell}=-p_{-(K-x)_+}$. We will denote by $V_{\zeta,K}(x_0,\l)$ and $p_{\zeta,K}$ (resp. $V(x_0,\l)$ and $p$) the value function of $\l>0$ units of $\zeta(S_T) 1_{S_T \leq K}$ (resp. $S_T$) and the associated bid reservation price. Similarly,  $V^{put}(x_0,\l,K)$ and $p^{put}$ will be respectively the value function of $\l$ units of  $-(K-S_T)_+$ and the ask reservation price of the put.\\
In the sequel, we will consider several market conditions. Situation 1 of Table \ref{parameters1} presents a market situation relatively close to what is commonly used in the literature (see for example \cite{ref1, ref4, ref5, HM15}), and that will be taken as a reference throughout the article. We also introduce two other market situations to observe the effect of a low initial price for the non-traded asset (see situation 2 of Table \ref{parameters1}) and a long time horizon (see situation 3 of Table \ref{parameters1}).  

\begin{table}[H]
\caption{Different market situations considered through the paper}
\begin{tabular*}{\textwidth}{|c|@{\extracolsep{\fill}}cccccccc|}
  \hline
    & $r$& $\lambda$ & T & $s_0$ & $\n$ & $\eta$ & $\mu$& $\sigma$  \\
  \hline
   Situation 1 & 0.1\% & 2  & 0.25  & 100 & 20\% & 30\% & 10\% & 20\% \\
   Situation 2 & 0.1\% & 20  & 0.3  & 1 & 35\% & 40\% & 10\% & 20\% \\
   Situation 3 & 0.1\% & 10  & 10  & 100 & 30\% & 30\% & 5\% & 10\% \\
      \hline
\end{tabular*}
\label{parameters1}
\end{table}

\section{Polynomial approximation}
\label{taylorapprox}

In this section, we search for some polynomial approximation for the reservation price that will be relevant in general market situations. 
Such an approach was already studied by many authors. Davis \cite{ref2} achieved an expansion with respect to $\sqrt{1-\rho^2}$ using ideas from Malliavin calculus, while Henderson \cite{ref1} focused on an expansion on $\lambda$ 
and found a similar expression. Later, Monoyios \cite{ref4,ref5} gave a higher order expansion on a parameter involving correlation, risk aversion and the number of claims for the ask reservation price.  It is worth noting that the Monoyios\rq{}s expansions are not Taylor expansions in the sense that some dependence upon $\rho$ remains present in the coefficients, which are defined using the minimal martingale measure. As a consequence, we strongly believe that the error order provided by Monoyios (see Theorem 3 in \cite{ref5}) may have been misestimated. Nevertheless, his approximation has proven to be reliable for the computation of the ask reservation price of a put option for a very small risk aversion (see Table 1 in \cite{ref5}).
We show in Figure \ref{monoyios} that this approximation remains correct for the computation of the bid reservation price of a stock for very small values of the risk aversion or very high correlation, but that this is no more the case elsewhere.

\begin{figure}[H]
\captionsetup[subfigure]{justification=centering}
  \centering
  \centering
\resizebox{0.75\textwidth}{!}{
  \begin{subfigure}[b]{0.49\linewidth}
    \includegraphics[width=\linewidth, height=7cm-0.5cm]{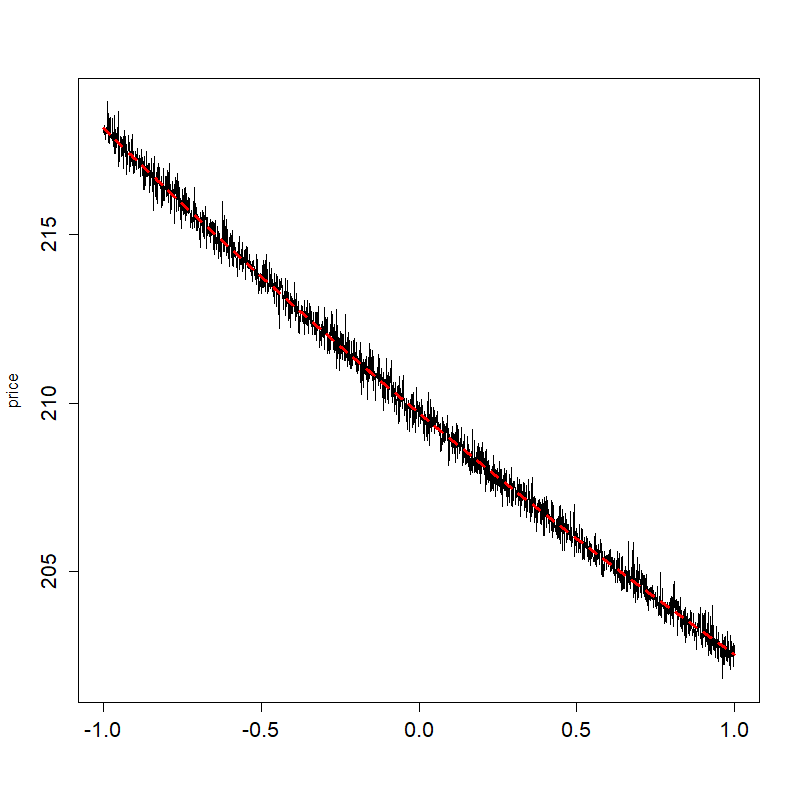}
     \caption*{$\gamma=0.001$}
  \end{subfigure}
  \begin{subfigure}[b]{0.49\linewidth}
    \includegraphics[width=\linewidth, height=7cm-0.5cm]{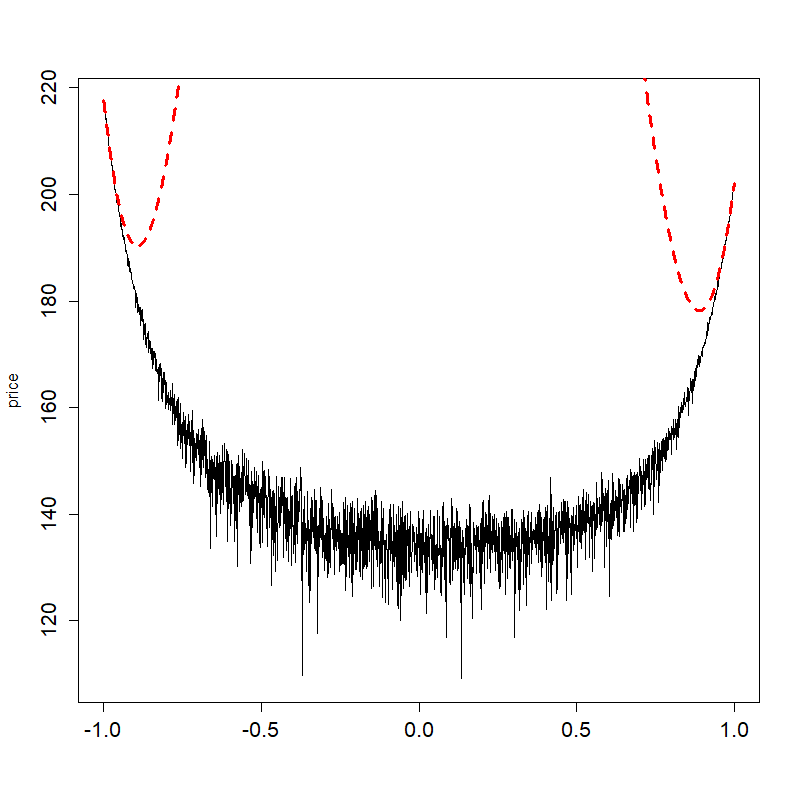}
    \caption*{$\gamma=0.5$}
  \end{subfigure}}
  \caption{In the market situation 1 specified in Table \ref{parameters1} with $\gamma=0.001$ and $\g=0.5,$ estimated bid reservation price of $\l S_T$ as function of the correlation  using the Monte Carlo method in black (solid line) and with the fourth-order approximation of Monoyios in red (dashed line).}
\label{monoyios}
\end{figure}
We propose in Appendix, a Taylor expansion of the bid reservation price of a non-traded stock, when the correlation is close to one, see Section \ref{apptaylor}. We show numerically in Figure \ref{figure_dl} that Taylor expansions remain an extremely bad approximation, when the correlation is not very close to one. 
\begin{figure}[H]
{\centering
\includegraphics[width=0.7\linewidth]{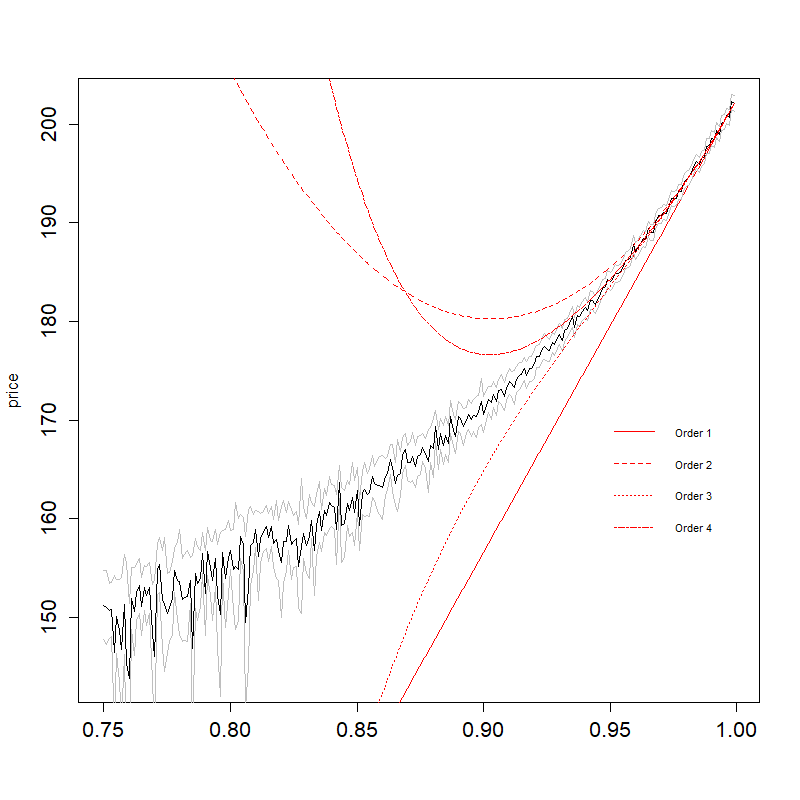}\par}
  \caption{In the market situation 1 of Table \ref{parameters1} for $\gamma=0.5$, the estimated bid reservation price of $\l S_T$ as a function of $\rho$ using the Monte Carlo method  is shown in black and Taylor expansions are displayed in red. 
The $99\%$ confidence intervals of the Monte Carlo estimate are in grey.}
\label{figure_dl}
\end{figure}
%
%
%
%
%
%

\section{Lambert approach}
\label{secdecomp}
The polynomial method described above only allows to obtain reliable results when the correlation is very close to $1$ (or $-1$) or when the risk aversion is very small. Moreover, there is no simple expression of the approximation error. Thus, we now provide another approach in the spirit of \cite{ref6}, where we split the reservation price into two parts. The first one, which we call the deterministic part for simplicity, can be computed without a Monte Carlo method. This is not the case for the second one, which we refer to as the random part.

The bid reservation price can be written as the logarithm of the Laplace transform of some function $f$ of a lognormal random variable and of an indicator function. Following \citep{ref6} and \citep{ref12}, we apply a modified version of the Laplace method (see \citep[Chapter 4]{ref11}) to compute the integral in the Laplace transform. In his thesis \citep{ref12}, Rojas-Nandayapa derived an approximation of the Laplace transform of a lognormal random variable $e^{\eta\sqrt{T}N}$ using a saddle point approximation. Later, Asmussen, Jensen and Rojas-Nandayapa derived rigorously in  \citep{ref6} the same approximation using a modified Laplace method. We generalize their approach to $f(e^{\eta\sqrt{T}N})1_{e^{\eta\sqrt{T}N}\leq \beta}$. We approximate $f(e^{\eta\sqrt{T}N})$ by a polynomial function and we get the maximum of $-k_\beta$  out of the integral, see \eqref{decompo_f} and \eqref{defkbeta}. This leads to a multiplicative decomposition of the Laplace transform and thus, an additive decomposition for the bid reservation price.
Based on this decomposition, we will propose in the next sections lower and an upper bounds for the reservation price and the value function which are not expressed using an expectation. We will see that those bounds are indeed accurate. We will provide interesting results on the lower bound as a function of $\r$ and $\gamma$. 
We will also propose an improved numerical method for the computation of the reservation price, based on an estimator with a lower variance, by performing the Monte Carlo method only on the random part, see Figures \ref{lambert_direct_stock} and \ref{lambert_direct_put}. 

We start with the decomposition of $L_{f,\beta},$ the Laplace transform of the random variable $f(e^{\eta\sqrt{T}N})1_{e^{\eta\sqrt{T}N} \leq \beta},$ where $\beta \in [0,+\infty]$ and $f$ is a measurable function that is bounded from below on $[0,\beta]$. Note that \cite{ref6} studies the case where $\beta=\infty$ and $f=id$. 
As $f$ is bounded from below on $[0,\beta]$, $L_{f,\beta}$ is well defined for all $\theta>0$ :
\begin{eqnarray*}
L_{f,\beta}(\theta)=\mathbb{E}\left[\exp\left(-\theta f\left(e^{\eta\sqrt{T}N}\right) 1_{e^{\eta\sqrt{T}N }\leq \beta} \right)\right].
\end{eqnarray*}
Let $K\in [0,+\infty]$ and $\zeta$ be a measurable function, that is bounded from below on $[0,K]$. Set 
$\hat{s}_0=s_0e^{(\nu-\eta\rho s_R-\frac{\eta^2}{2})T},$ $\hat{K}=K/\hat{s}_0$, $\hat{\theta}=\l\g(1-\r^2)$ and $\hat{\zeta}(x)=\zeta(\hat{s}_0 x)$ for all $x\geq 0$. There is a simple relation between $L_{\hat{\zeta},\hat{K}}$ and $p_{\zeta,K}$. Indeed, using \eqref{indiffdef_gen} with $h={\zeta}1_{[0, K]}$, we obtain that 
\begin{eqnarray}
p_{\zeta,K}&=& -\frac{e^{-rT}}{\gamma(1-\rho^2)}\ln\mathbb{E}\left[\exp\left(-\hat{\theta} \hat{\zeta}\left(e^{\eta\sqrt{T}N}\right)1_{e^{\eta\sqrt{T}N}\leq \hat{K}}\right) \right] 
 =  -\frac{e^{-rT}}{\gamma(1-\rho^2)}\ln 		L_{\hat{\zeta},\hat{K}}(\hat{\theta}).
\label{lien_PL}
\end{eqnarray}

We now propose a decomposition of $L_{f,\beta}(\theta)$. To do that, we employ a kind of Taylor expansion of $f$. At this stage, we do not specify $u, v \in \mathbb{R}$ with $v>0,$ as we use different values for $u$ and $v,$ whether $\beta=\infty$ or not. For all $\theta>0$, 
\begin{eqnarray}
L_{f,\beta}(\theta)  &= & \int_{\mathbb{R}}{\exp\left[-\theta\left(f\left(e^{\eta\sqrt{T}z}\right)-u-v e^{\eta\sqrt{T}z}\right)1_{z\leq \frac{\ln \beta }{\eta\sqrt{T}}} - k_{\beta}(z)\right]\frac{dz}{\sqrt{2\pi}}},
\label{decompo_f}
\end{eqnarray}
where for all $z\in \mathbb{R}$, we have set
$k(z)  =  \theta \left(u+v e^{\eta\sqrt{T}z}\right)+\frac{z^2}{2}$ and
\begin{eqnarray}
\label{defkbeta}
k_\beta(z)  & = & k(z)1_{z\leq \frac{\ln \beta}{\eta\sqrt{T}}}+\frac{z^2}{2}1_{z> \frac{\ln \beta}{\eta\sqrt{T}}}=\theta \left(u+v e^{\eta\sqrt{T}z}\right)1_{z\leq \frac{\ln \beta}{\eta\sqrt{T}}}+\frac{z^2}{2}. 
\end{eqnarray}
Note that $k_{+\infty}=k$. Moreover, $k_\beta$ is continuous if and only if $u=-v\beta$.

We seek to remove the integral in \eqref{decompo_f} of the maximum of $e^{-k_\beta}$. Thus, we aim to minimize $k_\beta$: to do that, we first study $k$. Obviously, $k$ is $C^{\infty}$, and $k'(z)=\theta v \eta \sqrt{T}e^{\eta\sqrt{T} z} +z$ for all $z\in\mathbb{R}$. As $v>0$, for all $z \geq 0$, we have $k'(z) \geq 0$. Assume now that $z<0$. Then, $k'(z) \ge 0$ if and only if $\theta v  \eta^2 T \ge - \eta\sqrt{T} z e^{-\eta\sqrt{T} z}.$
At this point, we need the inverse function of $x \mapsto xe^{x}$.  
\begin{definition}
The Lambert function $W:\left(-\frac{1}{e},+\infty\right) \to (-1,+\infty)$ is the inverse function of $x \mapsto xe^{x}.$ 
\end{definition}
We provide some properties of $W$ in the Lemma \ref{lemlamb} in the Appendix~:  in particular, $W>0$ on $(0,+\infty)$.
The computation of the minimum of $k_{\beta}$ is performed in Lemma \ref{variation_k} in the Appendix. This provides the following multiplicative decomposition of $L_{f,\beta}$, which extends  \cite[Proposition 2.1]{ref6}. 
\begin{theorem}
\label{propdecompo}
Let $\beta\in[0,+\infty]$, $\theta,v>0$ and $f$ be a measurable function that is bounded from below on $[0,\beta]$. Let $u$ be any real number if $\beta=+\infty$ and $u=-v\beta$ if $\beta<+\infty$. Moreover, suppose that 
\begin{eqnarray}
\frac{W(\theta v  \eta^2 T)}{\eta^{2}T}
+ \frac{W^2(\theta v  \eta^2 T)}{2\eta^{2} T}\leq \theta v\beta.
\label{domaine}
\end{eqnarray}
Then, 
$
L_{f,\beta}(\theta)={L}_\beta(\theta){I_{f,\beta}}(\theta),
$
where with the convention that $(-\infty)_+=0$ :
\begin{align*}
{L}_\beta(\theta) & =
\exp\left[-\left(\theta u +\frac{W(\theta v  \eta^2 T)}{\eta^{2}T}
+ \frac{W^2(\theta v  \eta^2 T)}{2\eta^{2} T}\right)\right] \\
I_{f,\beta}(\theta) & = 
\begin{multlined}[t][\textwidth-3cm]
 \mathbb{E}\left( \exp\left[-\theta\left(f\left(\frac{W(\theta v  \eta^2 T)}{\theta v\eta^2 T}e^{\eta\sqrt{T}N}\right)-u - \frac{W(\theta v \eta^2 T)}{\theta \eta^2 T}e^{\eta\sqrt{T}N}\right)1_{N\leq \frac{\ln(\beta)}{\eta\sqrt{T}}+\frac{W(\theta v \eta^2 T)}{\eta\sqrt{T}}}\right] \phi_\beta( N,\theta) \right) 
 \end{multlined}\\
\phi_\beta(y,\theta) & =  \exp\left[- \frac{W(\theta v  \eta^2 T)}{\eta^2 T}\left(e^{\eta\sqrt{T} y}-1-\eta\sqrt{T} y\right)\right]\exp\left[\left(\frac{W(\theta v  \eta^2 T)}{\eta^2 T}e^{\eta\sqrt{T}y}-\theta v \beta\right)_+\right],
\end{align*}
\end{theorem} 
\begin{proof}[Proof of Theorem \ref{propdecompo}] See Section \ref{proof_decompo_annexe}. $\;\square$ \end{proof}
With Theorem \ref{propdecompo}, we can derive a decomposition of the reservation price and of the value function for a derivative with payoff $ \zeta(S_T)1_{S_T\leq K}$. 
\begin{theorem}
\label{prop_dec_gen}
Let $K\in [0,+\infty]$ and $\zeta$ be a measurable function that is bounded from below on $[0,K]$. Set $\hat{s}_0=s_0e^{(\nu-\eta\rho s_R-\frac{\eta^2}{2})T}$ and let $v>0$ and $u=-v K/\hat{s}_0$ if $K<+\infty$ or $u$ be any real number if $K=+\infty$.
Moreover, suppose that
\begin{eqnarray}
\frac{W(\l\g(1-\r^2) v  \eta^2 T)}{\eta^{2}T}
+ \frac{W^2(\l\g(1-\r^2) v  \eta^2 T)}{2\eta^{2} T}\leq \frac{K}{\hat{s}_0}\l\g(1-\r^2) v. \label{condition}
\end{eqnarray}
Then,
\begin{eqnarray}
p_{\zeta,K}&=&D_{\zeta,K}+A_{\zeta,K} \label{eqequalg}\\
D_{\zeta,K} & = &\l e^{-rT}u+\frac{e^{-rT}}{\gamma(1-\rho^{2})}
\left(\frac{W(\l\g(1-\r^2) v  \eta^2 T)}{\eta^{2}T}
+ \frac{W^2(\l\g(1-\r^2) v  \eta^2 T)}{2\eta^{2} T}\right) \label{eqDg} \\
A_{\zeta,K} & =&  -\frac{e^{-rT}}{\gamma(1-\rho^{2})}\ln\left(I_{\hat{\zeta},\hat{K}}(\lambda\g(1-\r^2))\right)\label{eqAg},
\end{eqnarray}
where $I_{\hat{\zeta},\hat{K}}$ is defined in Theorem \ref{propdecompo} with 
$\hat{K}= K/\hat{s}_0$ and $\hat{\zeta}(x)=\zeta(\hat{s}_0 x)$ for all $x\geq 0.$

We call $D_{\zeta,K}$ the deterministic part of $p_{\zeta,K}$ and $A_{\zeta,K}$ its random part. 
Moreover, 
\begin{eqnarray}
V_{\zeta,K}\left(x_0,\lambda, \zeta 1_{[0,K]}\right)&=&V_{D,\zeta,K}\left(x_0,\lambda,\zeta 1_{[0,K]}\right)V_{A,\zeta,K}\left(\lambda,\zeta 1_{[0,K]}\right)
\label{dec_value}\\
V_{D,\zeta,K}\left(x_0,\lambda,\zeta 1_{[0,K]}\right)& = &-\frac{1}{\gamma}\exp\left[{-\gamma e^{rT} (x_0+D_{\zeta,K})-\frac{s_R^2}{2}T}\right]\label{hat_V} \\
V_{A,\zeta,K}\left(\lambda,\zeta 1_{[0,K]}\right)& = & \exp\left[{-\gamma  e^{rT} A_{\zeta,K}}\right].\label{alV}
\end{eqnarray}
\end{theorem}

\begin{proof}[Proof of Theorem \ref{prop_dec_gen}]
The decomposition of $p_{\zeta,K}$ is a direct consequence of $(\ref{lien_PL})$ and Theorem \ref{propdecompo} applied to $\theta=\hat{\theta} = \lambda \gamma (1-\rho^2)$ and $\beta = K/\hat{s}_0$. Then, the decomposition of the value function is obtained using \eqref{value_pe}. $\;\square$ 
\end{proof}
\begin{remark}
The deterministic approximation $D_{\zeta,K}$ of $p_{\zeta,K}$ has a good economic interpretation as it can be understood as a bid reservation price. Indeed, if we assume that $V_{D,\zeta,K}(x_0,\lambda,\zeta 1_{[0,K]})$ is a value function, then (\ref{hat_V}) and  $D_{\zeta,K}|_{\l=0}=0$ (see \eqref{eqDg}) show that 
\begin{eqnarray*}
V_{D,\zeta,K}\bigg(x_0-D_{\zeta,K},\lambda,\zeta 1_{[0,K]}\bigg)&=&-\frac{1}{\gamma}\exp\left(-\gamma x_0 e^{rT}-\frac{s_R^2}{2}T\right)=V_{D,\zeta,K}\left(x_0,0,\zeta 1_{[0,K]}\right).
\end{eqnarray*} 
\end{remark}
When $K=+\infty$, \eqref{condition} is always satisfied and we can apply Theorem \ref{prop_dec_gen} to the long stock problem without any constraint on the parameters. When $K<+\infty$ and $u=-v K/\hat{s}_0$, \eqref{condition} is equivalent to $D_{\zeta,K}\leq 0$. We give some insight about equation \eqref{condition} in Section \ref{secshorput} in the case of a short Put position on $S$. With the decomposition provided by (\ref{eqequalg}) and (\ref{dec_value}), we perform a Monte Carlo method only on the random part of the reservation price or of the value function (see (\ref{eqAg}) and (\ref{alV})). 
So, we need to choose $u$ and $v$ such that $A_{\zeta,K}$ is as small as possible, i.e., such that $I_{\hat{\zeta},\hat{K}}(\l\g(1-\r^2))$ is close to $1.$ 
This is possible since differences between functions and their Taylor expansion of order one in zero appear. 
 We call this the Lambert Monte Carlo (LMC) method. We show below that it is indeed numerically efficient: see Figures \ref{lambert_direct_stock} and \ref{lambert_direct_put}. 

\begin{algorithm}[h]
\caption{LMC method algorithm for the computation of the bid reservation price}
\begin{algorithmic}
\Require $n > 0$
\State \textbf{Step 1} : Choose $u$ and $v$ such that the hypothesis of Theorem \ref{prop_dec_gen} are satisfied 
\State \textbf{Step 2} : Generate a vector $U$ of $n$ i.i.d random variables of law $\mathcal{N}(0,1)$
\State \textbf{Step 3} : Set $V := \exp\bigg[-\theta\bigg(\hat{\zeta}\left(\frac{W(\theta v  \eta^2 T)}{\theta v\eta^2 T}e^{\eta\sqrt{T}U}\right)-u - \frac{W(\theta v \eta^2 T)}{\theta \eta^2 T}e^{\eta\sqrt{T}U}\bigg) 1_{U\leq \frac{\ln(\hat{K})}{\eta\sqrt{T}}+\frac{W(\theta v \eta^2 T)}{\eta\sqrt{T}}}\bigg] \phi_{\hat{K}}( U,\theta)$ with the parameters $\hat{\zeta}$, $\hat{K}$ and the function $\phi_{\hat{K}}$ specified in Theorem \ref{prop_dec_gen} and $\theta:= \lambda\gamma (1-\rho^2)$. 
\State \textbf{Step 4} : Compute the mean $m$ of $V$
\State \textbf{Step 5} : Compute the approximation of the random part $M:=-\frac{e^{-rT}}{\gamma(1-\rho^{2})}\ln\left(m\right)$
\State \textbf{Step 6} : Return the approximation of the bid reservation price $D_{\zeta,K}+M$.
\end{algorithmic}
\end{algorithm}

\section{Long stock position}
\label{seclongstock}
In this part, we focus on the case of a long stock position on the non-traded stock, which is, as already mentioned, of great importance in management science. 
After giving the decomposition of the reservation price, we study the quality of the deterministic part as an approximation of the reservation price. We also give an approximation of the optimal strategy and of the Greeks. We provide numerical illustrations of all these quantities. We also examine the problem of selecting the right hedging asset $P$. Indeed, there may be several possible choices for $P$ : taking the example of an index, the difference stocks that constitute this index are possible hedging assets. We show that choosing an asset $P$ with the correlation $\rho^*$ (the minimum of the deterministic part as a function of $\rho$) with $S$ provides a minimum reservation price, as well as a minimum cash position invested in $P$.

\subsection{Decomposition of the reservation price and of the value function}
The next theorem is a direct consequence of Theorem \ref{prop_dec_gen}.
Recall that $p$ is the bid reservation price of $\l>0$ units of $S$ and $V$ is the associated value function. 
\begin{theorem}
\label{prop_dec}
Let $\bar w=W(s_0\lambda\gamma \eta^2 Te^{(\nu-\eta\rho s_R -\frac{\eta^2}{2})T}  (1-\rho^2)).$
 Then, $p=D+A$, where
\begin{eqnarray}
D & = & \frac{e^{-rT}}{\gamma(1-\rho^{2})}
\left( \frac{\bar w}{\eta^2 T}+\frac{\bar w^2}{2\eta^2 T}\right)\label{eqD1} \\
A & = & -\frac{e^{-rT}}{\gamma(1-\rho^{2})}\ln \mathbb{E}\left(\exp\left[- \frac{\bar w}{\eta ^2 T}\left(e^{\eta\sqrt{T}N}-1-\eta\sqrt{T}N\right) \right]\right)\label{eqA}.\\\nonumber 
\end{eqnarray}
Moreover, $V(x_0,\l)=V_D(x_0,\lambda)V_A(\lambda),$ where 
\begin{eqnarray}
V_D(x_0,\lambda)& = &-\frac{1}{\gamma}\exp\left[{-\gamma e^{rT} (x_0+D)-\frac{s_R^2}{2}T}\right] \label{hat_D_value}\\
\nonumber
V_A(\lambda)& = & \exp\left[{-\gamma  e^{rT} A}\right]. 
\end{eqnarray}
\end{theorem}
\begin{remark}
\label{lower_bound_D}
As $e^{x}\geq x+1$ and $\bar w > 0$, we get that $A\geq 0,$ $V_A \leq 1$ and the deterministic part $D$ (resp. $V_D$) is a lower bound for $p$ (resp. $V$). 
\end{remark}
\begin{proof}[Proof of Theorem \ref{prop_dec}]
Recall that $\hat{s}_0=s_0e^{(\nu-\eta\rho s_R-\frac{\eta^2}{2})T}$ and let $\hat{\theta}=\l\g(1-\r^2)$.
We choose in Theorem \ref{prop_dec_gen} $K=+\infty$, $\zeta=id,$ $u=0$ and $v=\hat{s}_0.$ Then, (\ref{condition}) is obviously true. Moreover,  
 $\hat{K}=+\infty$  and $\hat{\zeta}(x)=\hat{s}_0x,$ for all $x\geq 0$. Then, for all $y\in \mathbb{R}$, 
\begin{align*}
\phi_{\hat{K}}(y,\hat{\theta})&=\exp\left[-\frac{\bar w}{\eta^2 T}\left(e^{\eta\sqrt{T}y}-1-\eta\sqrt{T}y\right)\right]\\
I_{\hat{\zeta},\hat{K}}(\hat \theta) & = \mathbb{E}\left(\phi_{\hat{K}}\left(N,\hat{\theta}\right)\; \exp\left[-\hat{\theta}\left(\hat{\zeta}\left(\frac{\bar w}{\hat{\theta} \hat{s}_0\eta^2 T}e^{\eta\sqrt{T}N}\right)- \frac{\bar w}{\hat{\theta} \eta^2 T}e^{\eta\sqrt{T}N}\right)\right]\right) = \mathbb{E}\left(\phi_{\hat{K}}(N,\hat{\theta})\right). \quad\hfill \square
\end{align*}
\end{proof}

\subsection{Deterministic approximations of $p$}
\label{subappp}

We have seen that $D$ provides a deterministic lower bound for the reservation price (see Remark \ref{lower_bound_D}). Using the Lambert function, we also provide a deterministic upper bound for $p$ called $G$ and study the quality of both bounds.
Let 
\begin{eqnarray}
G& = & \frac{e^{-rT}}{\gamma(1-\rho^{2})}
\left( \frac{\bar w}{\eta^2 T}e^{\frac{\eta^2 T}{2}}+\frac{\bar w^2}{2\eta^2 T}\right)\label{Grandg}\\
V_G(x_0,\l)&=& -\frac{1}{\gamma}\exp\left[{-\gamma e^{rT} (x_0+G)-\frac{s_R^2}{2}T}\right].
\nonumber
\end{eqnarray}
\begin{theorem} 
\label{Evry}
We have that $ D\leq p \leq G$ and that  $V_D(x_0,\l)\leq V(x_0,\l) \leq V_G(x_0,\l).$ Moreover, 
\begin{eqnarray}
\label{ineq1}
\frac{D}{p} & \geq & e^{-\frac{\eta^2 T}{2}}  \geq  \frac{1+\frac{\bar w}{2}}{e^{\frac{\eta^2T}{2}}+\frac{\bar w}{2}} \\
\label{ineq2}
 \frac{G}{p} & \leq & 1+\frac{\bar w e_2}{2\eta^2 T+\eta^2 T\bar w}\leq 1+ \frac{e_2}{\eta^2 T},
\end{eqnarray}
where $e_2=e^{2\eta^2 T}-2(1+\eta^2  T)e^{\frac{\eta^2 T}{2}}+\eta^2 T +1$.
\end{theorem}
\begin{proof}[Proof of Theorem \ref{Evry}]  See Section \ref{sub_52}. $\;\square$ \end{proof}
The bounds in \eqref{ineq1} easily apply to $A/p$ and show that $A$ is small relative to $p$ : $$0 \leq  \frac{A}{p}  \leq \frac{e^{\frac{\eta^2T}{2}}-1}{e^{\frac{\eta^2T}{2}}+\frac{\bar w}{2}}  \leq  1-e^{-\frac{\eta^2 T}{2}}.$$
Having bounds that only depend on $\eta$ and $T$ legitimizes the use of $D$ and $G$ as uniform approximations of the reservation price in terms of $\rho \in (-1,1),$ 
$\lambda>0$ or $\gamma>0$. For example, if $\eta^2 T\leq 0.04$, then $0.98 p \leq D \leq p$, and if $\eta^2 T \leq 0.145 $, $p\leq G \leq 1.02 p$.
When $s_0$ and $\lambda$ go to infinity, $\bar w$ goes to infinity (see Lemma \ref{lemlamb} in the appendix) and \eqref{ineq1} implies that $D$ goes to $p$. 
We will see that numerically, $D$ provides a better approximation when $s_0$ and $\lambda$ are large enough : see Figure \ref{s3}. 
In contrast, when $s_0$ and $\lambda$ go to $0$ or $|\rho|$ goes to $1$, $\bar w$ goes to 0 (see Lemma \ref{lemlamb} in the appendix) and \eqref{ineq2} implies that $G$ goes to $p$. So,   
$G$ is sharp for small values of $s_0$  and $\lambda$ or high values of $|\rho|$ : see Figure \ref{s2}.
The limits in terms of $\rho$ and $\gamma$ are less straightforward. 
Let $d, \,g\; :\; \rho\in (-1,1) \to \mathbb{R}$ be such that $D=d(\rho)$ and $G=g(\rho).$ Let $\widetilde{d}\; :\; \gamma\in (0,\infty) \to \mathbb{R}$ be such that $D=\widetilde{d}(\gamma)$. 
\begin{proposition}
\label{propa}
The functions $d$, $\widetilde{d}$ and $g$ are $C^1$, positive and bounded. Moreover,
\begin{eqnarray}
\label{eqd0}
\underset{\rho\rightarrow 1^-}{\lim} d(\rho)  =  \l e^{-rT} s_0 e^{\left(\nu-\eta s_R -\frac{\eta^2}{2}\right)T} &  &  \underset{\rho\rightarrow -1^+}{\lim}d(\rho) =  \l e^{-rT} s_0 e^{\left(\nu+\eta s_R -\frac{\eta^2}{2}\right)T} \\
\label{eqg0}
\lim_{\rho\to 1^-}g(\rho) = \l e^{-r T}s_0 e^{\left(\nu-\eta s_R\right)T} &  &
\lim_{\rho\to -1^+}g(\rho)= \l e^{-r T}s_0 e^{\left(\nu+\eta s_R\right)T} \\
\lim_{\gamma\to 0^+}\widetilde d(\gamma) = \l e^{-rT} s_0 e^{\left(\nu-\eta\rho s_R-\frac{\eta^2}{2}\right)T}  &  &
\lim_{\gamma\to +\infty}\widetilde d(\gamma) = 0. \label{lim_inf_dg}
\end{eqnarray} 
\end{proposition} 
\begin{proof}[Proof of Proposition \ref{propa}]  See Section \ref{sub_52}. $\;\square$
\end{proof}

Note that $\hat{p} = \lambda e^{-rT} s_0e^{(\nu-\eta s_R)T}$ is an approximation at order $0$ of $p$ near $\rho=1^-$ (see \eqref{taylor}). 
Thus, unlike $d$, the upper bound $g$ is not biased for high positive correlation. 
Note that (\ref{eqd0}) shows that $\lim_{\rho \to 1^-} d(\rho)/\hat{p}= e^{-\frac{\eta^2}{2}T}.$ 
Thus, small values of $\eta$ and $T$ reduce the bias between $d$ and $p$ for high positive correlations.\\ 
When the risk aversion goes to $+\infty$, $D$ goes to zero, which is the subhedging price of $\l S_T.$ When the risk aversion goes to $0,$ $D$ goes to $\widetilde p e^{-\frac{\eta^2}{2}T}$, where $\widetilde p$ is the expectation of $\lambda S_T e^{-rT}$ under the minimal martingale measure, i.e. the minimal variance price of the claim $\lambda S_T$, see \citep{DR91} and \citep{S92}. Recall that the minimal variance price is the initial wealth that minimizes the quadratic hedging error under
the historical probability.\\

We compute now some Greeks like quantities for the lower bound $d$. We are particularly interested in the Cega $\Delta_{\rho}$ i.e. the sensitivity of the long stock position to a change in the correlation. We also compute the sensitivity $\Delta_{\g}$ with respect to $\gamma$. 


\begin{proposition}
We have that
\label{sensi}
\begin{eqnarray*}
\Delta_\rho &=&\frac{\partial d}{\partial \rho}  =  
-\frac{ e^{-rT}}{\eta^2 T \gamma(1-\rho^2)}\left(  \eta s_R T  \bar{w} - \frac{\rho}{1-\rho^2} \bar{w}^2\right)\\
\Delta_\gamma &=&\frac{\partial d}{\partial \gamma}
  =  -\frac{ e^{-rT}}{2 \eta^2 T\gamma^2(1-\rho^2)}\bar{w}^2.
\end{eqnarray*}
 \end{proposition}
\begin{proof}[Proof of Proposition \ref{sensi}]  See Section \ref{sub_52}. $\;\square$ \end{proof}

Note that the dynamic version of the (true) sensitivity of the long stock position to a change in the correlation is present in the expression of the optimal strategy, see (\ref{sup_strat}).

{\subsection{Numerical applications for the reservation price}}
\label{secnum}
We now provide numerical simulations to illustrate the performance of the Lambert Monte Carlo (LMC) method in comparison with the Direct Monte Carlo (DMC) method and also the quality of our deterministic approximations $D$ and $G$ of the reservation price $p$. We also plot the correlation $\rho^*$ that minimizes $D$, which will be studied in Section \ref{subvard} below.
\begin{figure}[h]
\captionsetup[subfigure]{justification=centering}
  \centering
  \begin{subfigure}[b]{0.49\linewidth}
    \includegraphics[width=\linewidth]{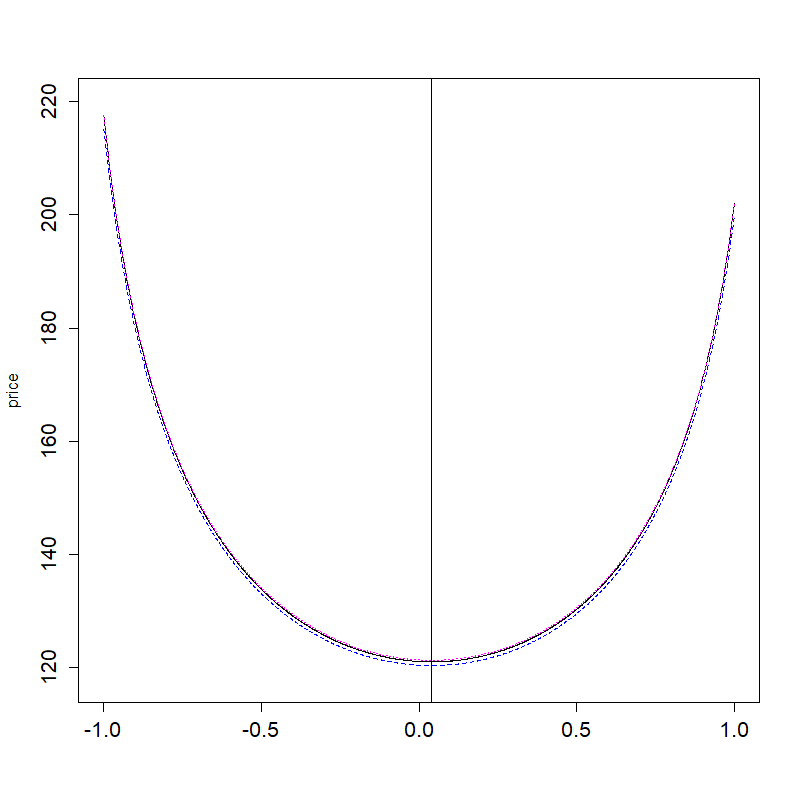}
     \caption*{Lambert Monte Carlo method}
  \end{subfigure}
  \begin{subfigure}[b]{0.49\linewidth}
    \includegraphics[width=\linewidth]{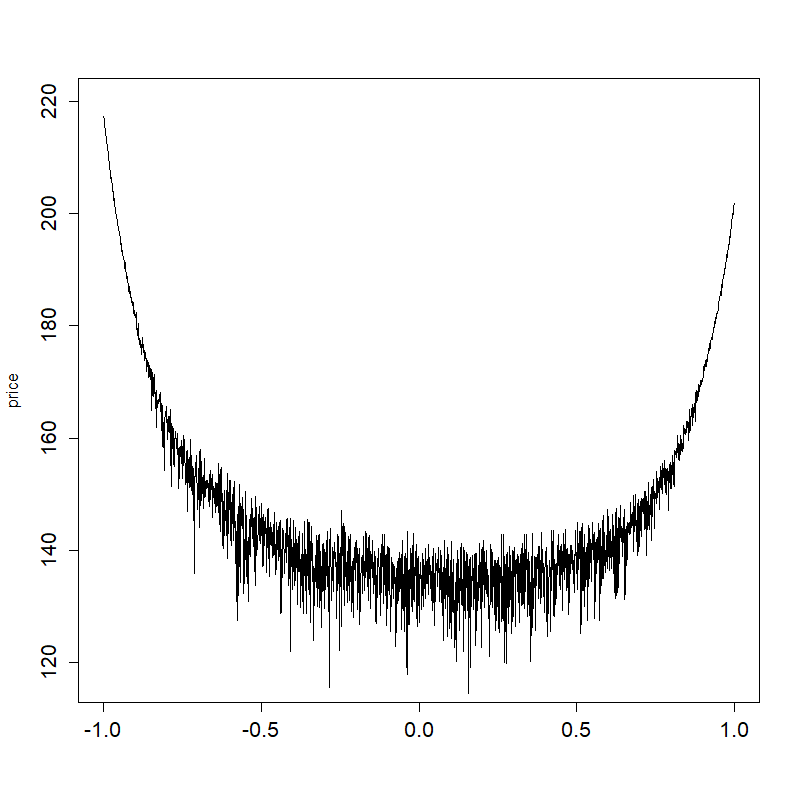}
    \caption*{Direct Monte Carlo method}
  \end{subfigure}
  \caption{In the market situation 1 of Table \ref{parameters1} with $\gamma=0.5$, estimated bid reservation price of $\l S_T$ as a function of $\rho$ in black, lower deterministic approximation $d$ in blue (dashed line) and upper deterministic approximation $g$ in purple (dotted line) with $10^4$ simulations per point. The vertical lines indicate the position of the minimum of $D$.}
\label{lambert_direct_stock}
\end{figure}

\begin{table}[h]
\centering

\begin{tabular*}{0.71\textwidth}{
| c | c | c | c | c |}
  \hline
     \multicolumn{1}{|c|}{\textbf{$\lambda$}} & \multicolumn{1}{|c|}{\textbf{LMC method}} & \multicolumn{1}{|c|}{\textbf{DMC method}} & \multicolumn{1}{|c|}{\textbf{Lower bound $D$}}  & \multicolumn{1}{|c|}{\textbf{Upper bound $G$}}   \\
  \hline
   $0.01$\; & $101.8$ & $101.8$ & $100.7$  & $101.8$ \\
   $0.1$ & $100.0$ & $100.0$ & $98.9$ &$100.0$ \\
   $1$ & $86.6$ & $86.6$ & $85.8$ &$86.7$ \\
   $10$ & $48.4$ & $56.6$ &$48.1$& $48.5$\\
   $20$ & $36.3$ & $53.7$ &$36.2$& $36.4$\\
   \hline
\end{tabular*}
\caption{In the market situation 1 of Table \ref{parameters1} with $\gamma=0.5$ and $\rho = 0.8$, estimations of the reservation price of $\lambda$ units of stock divided by $\lambda$ (per unit price) using both methods with $10^6$ simulations for several values of $\lambda$. The fourth and fifth columns are the deterministic bounds $D$ and $G$ of the (true) per unit price.}
\label{estimation_lbd}
\end{table}

We see in Figure \ref{lambert_direct_stock} that the LMC estimator of the bid reservation price of $\l S_T$ does not suffer from high variance, as is the case for the DMC estimator. For example, in situation 1 of Table \ref{parameters1} with $\gamma=0.5$ and $\rho=0$, the variance of the DMC (for $10^4$ simulations per point) is 5.66 and that of the LMC method is 0.017. We see that $D$ and $G$ are indeed good approximations of $p$. Moreover, the quality of the estimates given with the LMC method is much less sensitive to parameter variation than the one given by the DMC method (for a fixed number of simulations) as showed in Table \ref{estimation_lbd}. Indeed, for large values of $\lambda,$ the DMC method provides wrong estimates, while the estimations performed with the LMC method remain correct. 
 This can be explained by the fact that the DMC estimator needs more than $10^6$ simulations to be pertinent for higher value of $\lambda$. 
On the other hand, the quality of the LMC estimator is relatively stable for all values of $\lambda$ for $10^6$ simulations and this estimator is very close to the real price. Note that this abnormality can also be observed when working with high risk aversion $\gamma$ or with high initial price $s_0$.\\ 

\begin{figure}[h]
\captionsetup[subfigure]{justification=centering}
  \centering
  \begin{subfigure}[b]{0.35\linewidth}
    \includegraphics[width=\linewidth]{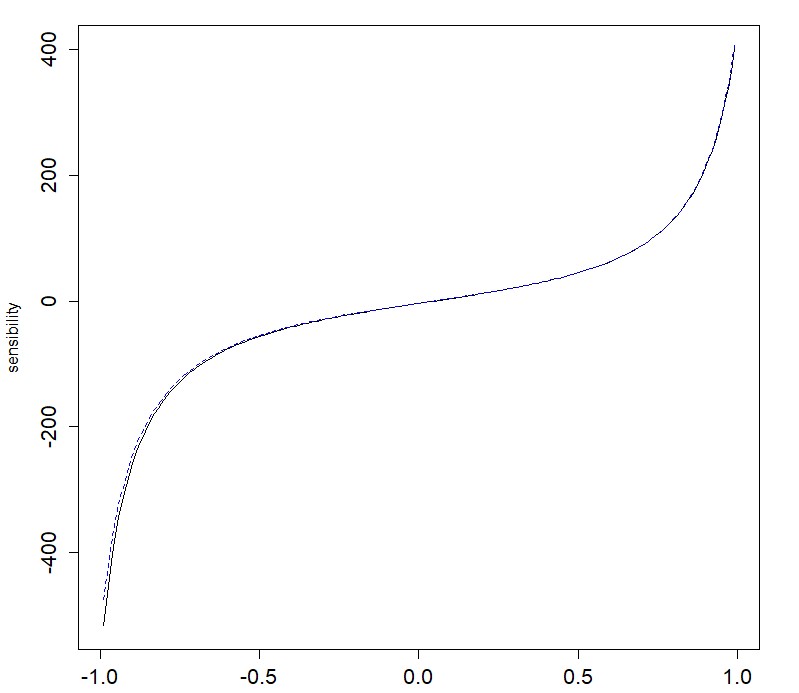}
    \caption*{Sensitivity on $\rho$ : $\frac{\partial p}{\partial \rho}$ and $\Delta_\rho$}
  \end{subfigure}
  \centering
  \begin{subfigure}[b]{0.35\linewidth}
    \includegraphics[width=\linewidth]{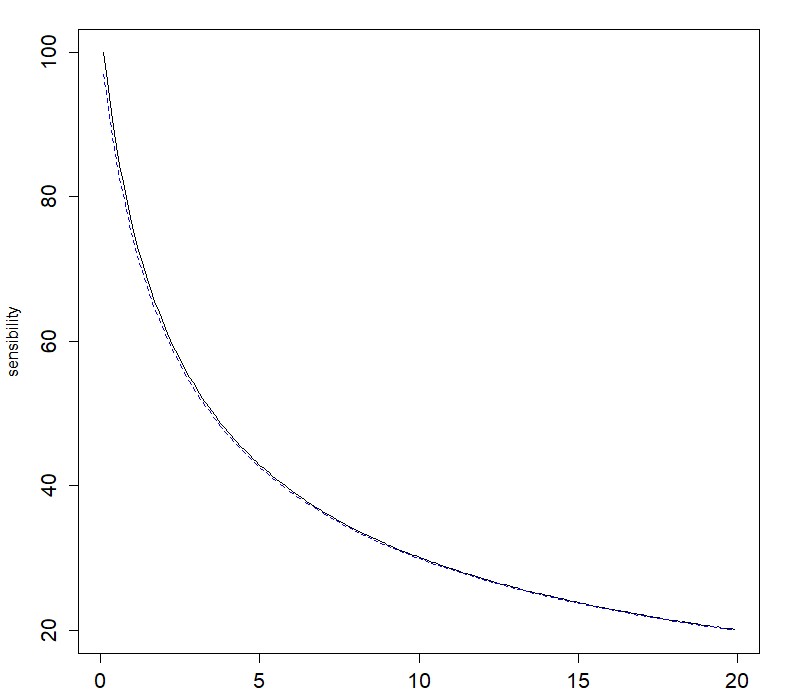}
     \caption*{Sensitivity on $\lambda$ : $\frac{\partial p}{\partial \lambda}$ and $\Delta_\lambda$}
  \end{subfigure}
  \caption{In the market situation 1 of Table \ref{parameters1} with $\gamma=0.5$ and $\rho=0.8$, (estimated) sensitivity of the price $p$ in black and (exact) sensitivity of the lower deterministic approximation $d$ in blue, with respect to $\rho$ and $\lambda$.}
  \label{sensi_fig}
\end{figure}


We now discuss the sensitivity of the price $p$ with respect to $\rho$ and $\lambda$. We show numerically that the sensitivity of $d$ can be very close to the one of $p$ (with respect to $\rho$ and $\l$), see Figure \ref{sensi_fig}. The two curves are indeed very close even when $d$ is not necessarily a good approximation of $p$. For example, when $\rho=-1$, the relative error between $\frac{\partial p}{\partial \rho}$ and $\Delta_\rho$ is equal to $7\%$ and when $\l=0$, the relative error between $\frac{\partial p}{\partial \lambda}$ and $\Delta_\lambda$ is equal to $3\%$.\\

We now propose other simulations in the situations 2 and 3 of Table \ref{parameters1} in order to illustrate the quality of $D$ and $G$. 
We see in Figure \ref{s2} that the upper deterministic approximation $G$ performs better when $\l s_0$ and $T$ are small (situation 2), while when $T$ and $\l s_0$ are large (situation 3), the lower bound $D$ is clearly more accurate. For the numerical values of situation 3 of Table \ref{parameters1}, the DMC does not work because of rounding issues. Moreover, note that the minimizer of the reservation price in $\rho$ is indeed close to the argminimum $\rho^*$ of the deterministic part $d$, even in situation 3, where $D$ has no reason to be a good approximation for $p$.\\ 
\begin{figure}[h]
\captionsetup[subfigure]{justification=centering}
  \centering
  \begin{subfigure}[b]{0.33\linewidth}
    \includegraphics[width=\linewidth]{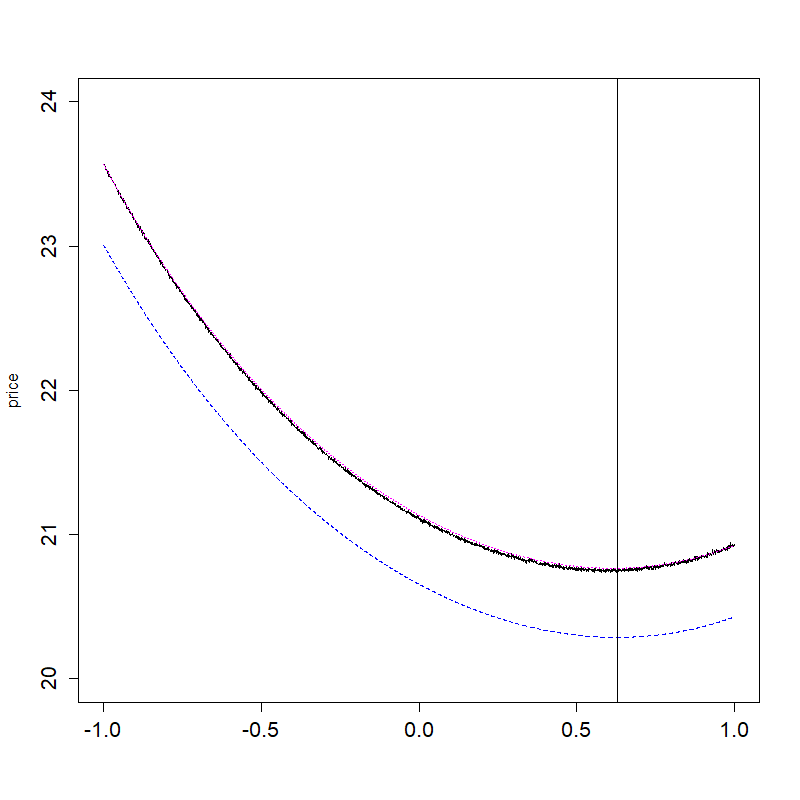}
    \caption*{LMC method,\\ Market Situation 2, $\gamma=0.1$}
  \end{subfigure}
  \centering
  \begin{subfigure}[b]{0.33\linewidth}
    \includegraphics[width=\linewidth]{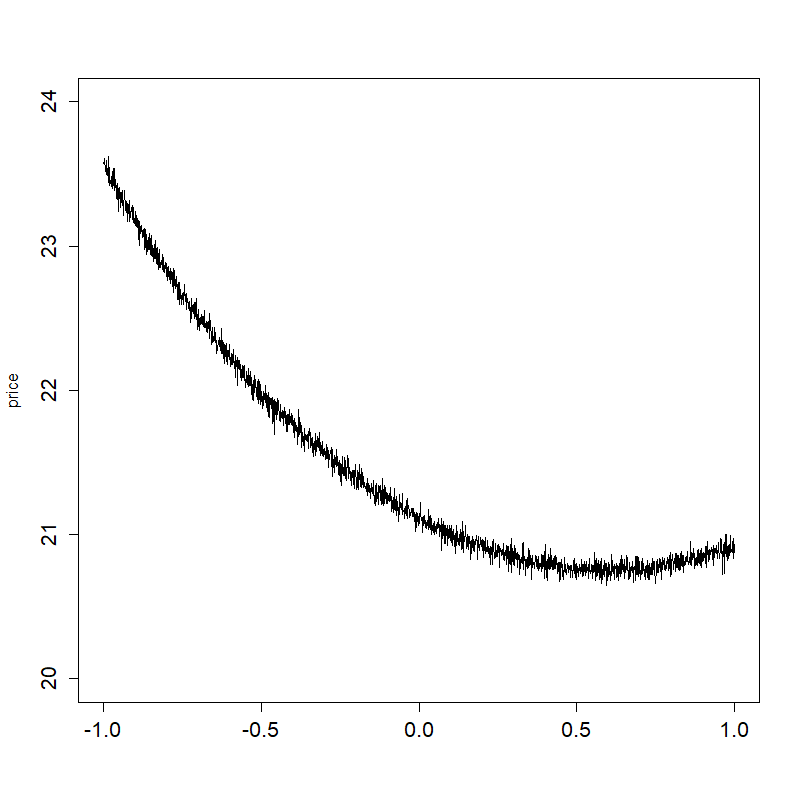}
     \caption*{DMC method,\\ Market Situation 2, $\gamma=0.1$}
  \end{subfigure}
    \begin{subfigure}[b]{0.33\linewidth}
    \includegraphics[width=\linewidth]{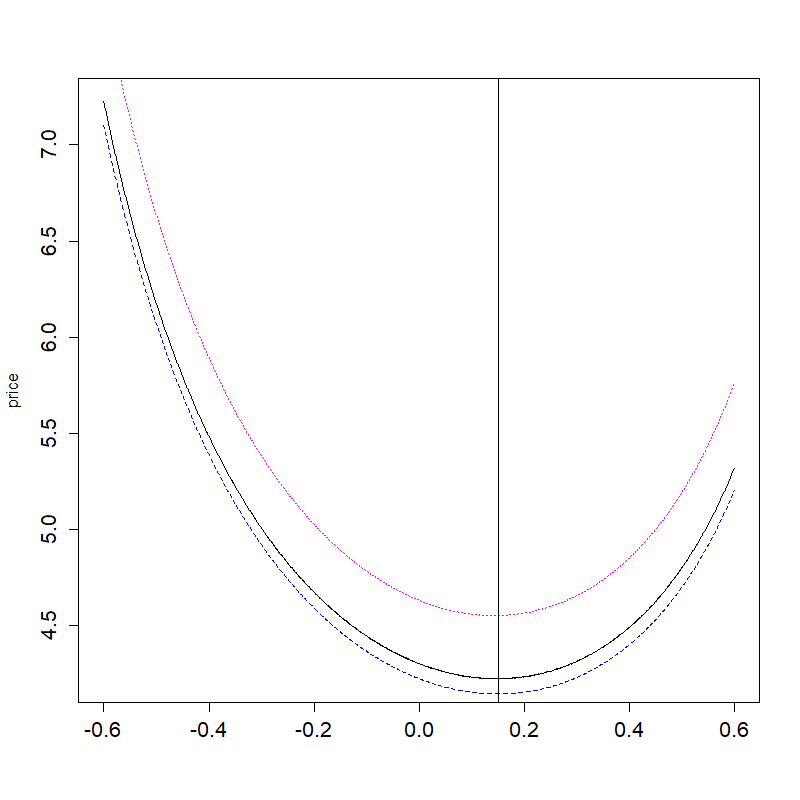}
     \caption*{LMC method,\\ Market Situation 3, $\gamma=15$}
  \end{subfigure}
  \caption{Estimated bid reservation price of $\l S_T$ as a function of $\rho$ in black,  lower deterministic approximation $d$ in blue (dashed line) and upper deterministic approximation $g$ in purple (dotted line) with $10^4$ simulations per point. The vertical lines indicate the position of the minimum of $D$.}
  \label{s2}
  \label{s3}
\captionsetup[subfigure]{justification=centering}
  \centering

\end{figure}

We now comment on the computational time of both the LMC and the DMC methods. In the market situation 1 of Table \ref{parameters1} with $\rho=0.4$ and $\gamma=0.5$, the execution time\footnote{CPU specification : Intel Core i5 - 7300U 2.60GHz, RAM specification : 8GB DDR3} of the LMC (resp. DMC) method for the computation of the reservation price and the Confidence Interval (CI) with $10^4$ simulations is 2.72 (resp. 2.10) milliseconds. For $10^6$ simulations, it is 229.14 (resp. 214.91) milliseconds. So, both methods take roughly the same time to produce a given price and CI. Indeed, most of the execution time of the LMC and the DMC methods comes from the simulation of random variables and the part  dedicated to the computation of the Lambert function is negligible.\\ 

In contrast to the LMC method, the number of simulations required for the DMC method to obtain a satisfactory CI can be very large. This phenomenon is illustrated in Table \ref{number_n}. 
We see that for some values of the correlation, the DMC method needs $10^8$ simulations to deliver an appropriate CI  while $10^2$ simulations are enough for the LMC method.  
Note that in the situation 1 of Table \ref{parameters1} and with our computer, the DMC method fails to deliver  a CI of  length smaller than $1$ for $\rho \in \{-0.6, -0.4, -0.2, 0.2,0.4, 0.6 \}$ as it needs more that $10^8$ simulations.

\begin{table}[H]
\centering
\centering
\resizebox{0.64\textwidth}{!}{
\begin{tabular*}{0.71\textwidth}{
| c | c | c | c | c | c | c |}
\hline
\multicolumn{1}{|c|}{} &
  \multicolumn{2}{|c|}{\textbf{LMC method}} & \multicolumn{2}{|c|}{\textbf{DMC method}} & \multicolumn{2}{|c|}{\textbf{Bounds}} \\
  \hline
     \multicolumn{1}{|c|}{\textbf{$\rho$}} & \multicolumn{1}{|c|}{\textbf{$n$}} & \multicolumn{1}{|c|}{\textbf{Time}} & \multicolumn{1}{|c|}{\textbf{$n$}} & \multicolumn{1}{|c|}{\textbf{Time}} & \multicolumn{1}{|c|}{\textbf{Lower bound $D$}} & \multicolumn{1}{|c|}{\textbf{Upper bound $G$}} \\
  \hline
   $-0.9$\; & $2$ & $0.13$ & $6$ & $179.23$ & 179.97 & 181.72\\
   $-0.8$ & $2$ & $0.16$   & $7$ & $1795.59$ & 160.67 & 162.11 \\
   $-0.6$ & $2$ & $0.16$ & $5$ & $16.70$ & 139.42 & 140.57 \\
   $-0.4$ & $2$ & 0.15 & $8$ & 18020.96 & 128.39 & 129.40 \\
   $-0.2$ & $2$ & 0.15 & $6$& 170.80 & 122.62 & 123.57\\
   $0$ & $2$ & 0.16 & $6$ & 172.85 & 120.44 & 121.37\\
   $0.2$ & $2$ & 0.16 & $8$ & 16974.27 & 121.37 & 122.32\\
   $0.4$ & $2$ & 0.17 & $6$ & 170.00 & 125.73 & 126.73 \\
   $0.6$ & $2$ & 0.15 & $8$ & 18476.06 & 134.93 & 136.06 \\
   $0.8$ & $2$ & 0.19 & $6$ & 172.26 & 153.23 & 154.62\\
   $0.9$ & $2$ & 0.15 & $5$ & 17.38 & 169.91 & 171.57\\
   \hline
\end{tabular*}}
\caption{Numbers of simulations ($10^n$) needed to obtain a CI of length smaller than $5$ for the reservation price of 2 units of stock and associated computational times  in milliseconds. This was computed for several value of $\rho$ in the market situation 1 of Table \ref{parameters1}. The last two columns are the (deterministic) bounds of the (true) price. CPU specification : Intel(R) Core(TM) i5-4690K 3.5GHz, RAM specification : 8GB DDR3.}
\label{number_n}
\end{table}  

We now compute the deterministic bounds for $D/p$ and $G/p$ of Theorem  \ref{Evry} in order to quantify the precision of our deterministic approximations of $p$ and to determine how sharp they are. We start with the lower bounds of $D/p$. In the market situation 1 of Table \ref{parameters1}, we have that $\exp{(-{\eta^2 T}/{2})}=0.9888$. The other bound in (\ref{ineq1}) 
is sharper (for $\gamma \in \{0.5,4,15\}$ and $\rho \in \{-0.8,-0.4,0,0.4,0.8\}$, it is between 0.9910 and 0.9956). 
Thus, the reservation price is well-explained by $D$. So in this situation, the two bounds for $D/p$ are indeed very close. This is not the case in the situation 3 of Table \ref{parameters1}. 
Indeed, we have $\exp{(-{\eta^2 T}/{2})}=0.6376,$ while the other bound in (\ref{ineq1}) is much tighter, as seen in Table \ref{sdev}. 
\begin{table}[H]
\begin{tabular*}{\textwidth}{|l@{\extracolsep{\fill}}|ccccc|}
  \hline
    & $\rho=-0.8$ & $\rho=-0.4$ & $\rho=0$ & $\rho=0.4$ & $\rho=0.8$ \\
  \hline
   $\gamma=0.5$ &0.8735 &0.8735 &0.8688 &0.8567 &0.8301 \\
   $\gamma=4$ & 0.9936 &  0.9944 & 0.9946  &0.9944 &0.9935\\
   $\gamma=15$ &0.9949 &0.9956 &0.9956 &0.9955 &0.9948 \\
   \hline
\end{tabular*}
\caption{Lower bound in (\ref{ineq1}) in the market situation 3 of Table \ref{parameters1}.}
\label{sdev}
\end{table}
We now compute the upper bounds of $G/p$ in \eqref{ineq2}. In the market situation 1 of Table \ref{parameters1}, we have that $1+e_2/\eta^2 T=1.017$. The other bound is between 1.0035 and 1.0106 for $\gamma \in \{0.5,4,15\}$ and $\rho \in \{-0.8,-0.4,0,0.4,0.8\}.$ 
The bounds are again very close, and the reservation price is also well explained by $G$. 
In the market situation 3 of Table \ref{parameters1}, we find that $1+e_2/\eta^2 T=3.2112$ and the values of the other upper bound in \eqref{ineq2} are greater than $2,$ even though $G$ is still a correct approximation of $p$ (see Figure \ref{s3}).

As a conclusion to this section, we have seen that $D$ and $G$ are reliable approximations of the  reservation price and are obtained without any computation time.


\subsection{Optimal strategy and value function}
\label{subsec_strategie}

We focus now on the optimal strategy $(\Pi^{*,\l}(t,S_t))_{0\leq t\leq T}$ for (\ref{valeur}), i.e., the strategy that maximizes the expected utility in (\ref{valeur}), when $h= id$. As we consider a bid reservation price instead of an ask reservation price (see (\ref{indiffsell_gen})), we need to slightly adapt the arguments of Monoyios (see (29) in \cite{ref5}). For all $t\in [0,T]$ and $s\geq 0$, we obtain that 
\begin{eqnarray}
\Pi^{*,\lambda}(t,s)&=&e^{-r(T-t)}\frac{s_R}{\gamma\sigma}-\frac{\eta\rho}{\sigma}s\frac{\partial p_t}{\partial s}(\rho,s),
\label{sup_strat}
\end{eqnarray} 
where $p_t$ is the dynamic version of $p$, i.e., $V_t(x-p_t,\l,id)=V_t(x,0,id),$ where $V_t(x,\l,id)$ is the value function at $t$ if $S_t=s$ and $X_t=x.$ For all $t\in [0,T]$ and $s\geq 0$, as for (\ref{indiffdef_gen}), we determine that 
\begin{eqnarray}
\label{dynamical_price}
p_t(\rho, s)=-\frac{e^{-r(T-t)}}{\gamma(1-\rho^2)}\ln\mathbb{E}\left(\exp\left[-\lambda\gamma(1-\rho^2)se^{\left(\nu-\eta \rho s_R-\frac{\eta^2}{2}\right)(T-t)+\eta\sqrt{T-t}N}\right]\right).
\end{eqnarray}
The computation of the optimal strategy   $\Pi^{*,\lambda}$ involves the partial derivative of $p_t.$ Therefore, numerical issues may appear because the DMC method performs badly even for $p$: see Figure \ref{lambert_direct_stock}. Monoyios (see Corollary 1 in \cite{ref4} or Section 4.1.1 in \cite{ref5}) addresses this issue by computing the series expansion of the partial derivative in (\ref{sup_strat}). However, we have seen in Section \ref{taylorapprox} that the series expansion can be extremely imprecise for $p,$ especially when $\gamma$ is not very small and $|\rho|$ is not near 1. Thus, we propose to use in \eqref{sup_strat} a dynamic version of our decomposition of the reservation price: 
\begin{eqnarray*}
\Pi^{*,\lambda}(t,s)&=&e^{-r(T-t)}\frac{s_R}{\gamma\sigma}-\frac{\eta\rho}{\sigma}s\frac{\partial d_t}{\partial s}(\rho,s)-\frac{\eta\rho}{\sigma}s\frac{\partial a_t}{\partial s}(\rho,s),
\end{eqnarray*}
where for all $s>0$ and $\rho\in (-1,1)$, $w_t(\rho,s)  =  W(s \lambda \gamma\eta^2 (T-t)e^{(\nu -\eta \rho s_R-\frac{\eta^2}{2})(T-t)} (1-\rho^2))$ and
\begin{align}
d_t(\rho,s)& = 
 \frac{  e^{-r(T-t)}}{\gamma\eta^2 (T-t)(1-\rho^2)} \left( w_t (\r,s)+\frac{w_t^2 (\r,s)}{2}\right) \label{dynamical_d}\\
a_t(\rho,s)& =   -\frac{ e^{-r(T-t)}}{\gamma(1-\rho^2)}\ln \mathbb{E}\left(\exp\left(- \frac{w_t(\rho,s)}{\eta ^2 (T-t)}\left(e^{\eta\sqrt{T-t}N}-1-\eta\sqrt{T-t}N\right) \right)\right)\nonumber.
\end{align}
Using (\ref{eqlambderiv}) in the Appendix, we determine that
\begin{eqnarray}
\frac{\partial {w}_t}{\partial s}(\rho,s)&=&\frac{{w}_t(\rho,s)}{s(1+{w}_t(\rho,s))}\nonumber\\
\frac{\partial d_t}{\partial s}(\rho,s)&=&  \frac{ e^{-r(T-t)}}{\gamma(1-\rho^2)\eta^2 (T-t)}\frac{\partial {w}_t}{\partial s}(\rho,s)(1+{w}_t(\rho,s)) 
= \frac{e^{-r(T-t)}}{\gamma(1-\rho^2)\eta^2 (T-t)s}{w}_t(\rho,s).
\label{delta_d_t}
\end{eqnarray}
Thus, we propose to approximate $({\Pi}^{*,\lambda}(t,S_t))_{0\leq t\leq T}$ with the deterministic strategy $({\Pi}^{D,\lambda}(t,S_t))_{0\leq t\leq T}$ defined for all $t\in [0,T]$ and $s>0$ by
\begin{eqnarray}
{\Pi}^{D,\lambda}(t,s)&=&e^{-r(T-t)}\frac{s_R}{\gamma\sigma}-\frac{\eta\rho}{\sigma}s\frac{\partial d_t}{\partial s}(\rho,s)= e^{-r(T-t)}\left(\frac{s_R}{\gamma\sigma}-\frac{\rho\; {w}_t(\r,s)}{\sigma\eta\g (1-\r^2)(T-t)}\right).
\label{new_strat}
\end{eqnarray}  
This strategy avoids taking the partial derivative and the related numerical issues. We do not provide any theoretical results regarding the quality of this approximation, but we show numerically in Figure \ref{strategie} that both strategies are indeed close. Note that in the market situation 1 of Table \ref{parameters1}, when 
 $x_0=80$ and $\gamma=0.5$, the value function $V$ is very small as $V\leq V_G=-2.42.10^{-45}.$ So, providing numerical approximations for $V$ is not meaningful. 
 We thus turn to the market situation 2 of Table \ref{parameters1} with $\gamma=0.1$ and $x_0=0.$ 
\begin{figure}[h]
\captionsetup[subfigure]{justification=centering}
  \centering
  \begin{subfigure}[b]{0.45\linewidth}
    \includegraphics[width=\linewidth]{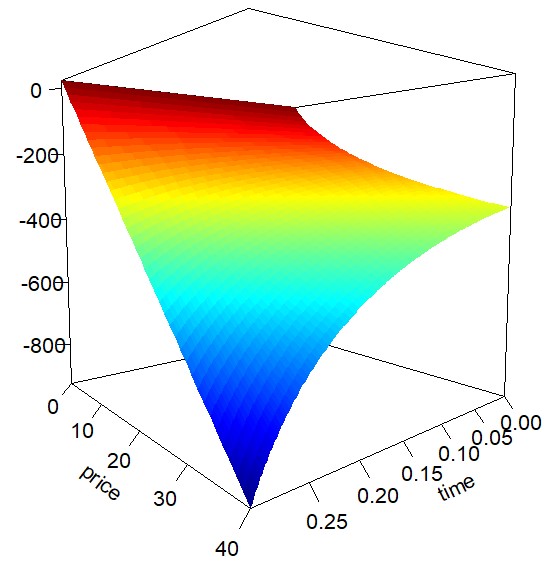}
     \caption*{Estimated optimal strategy}
  \end{subfigure}
 \begin{subfigure}[b]{0.45\linewidth}
    \includegraphics[width=\linewidth]{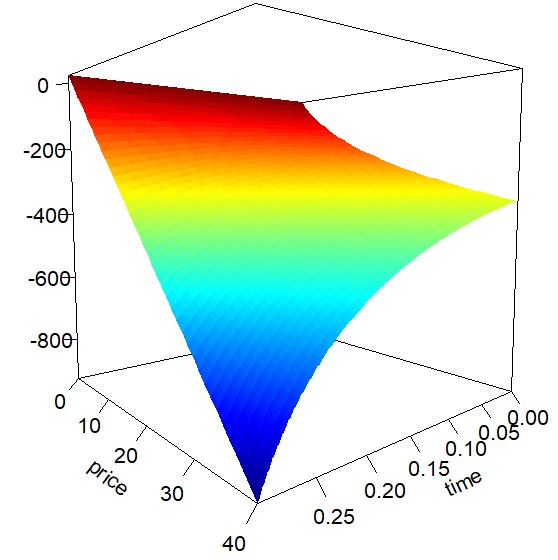}
    \caption*{Deterministic strategy}
  \end{subfigure}
  \caption{Graph of both strategies in the market situation 2 in Table \ref{parameters1} with $\gamma=0.1$, $x_0=0$ and $\rho=0.6$. For the computation of the estimated optimal strategy, we use the LMC method with $10^4$ simulations per point.}
 \label{strategie}
\end{figure}
It is also interesting to compare  $E_D=\mathbb{E}(U(X_T^{x_0,{\Pi}^{D,\lambda}}+\l S_T))$ with the value function $V=V(x_0,\lambda,id)$ defined in \eqref{def_value}. 


\begin{table}[H]
\centering
\resizebox{1\textwidth}{!}{
\begin{tabular*}{1.1\textwidth}{|l@{\extracolsep{\fill}}|cccccc|}
  \hline
    & $\rho=-0.8$ & $\rho=-0.5$ & $\rho=-0.2$ & $\rho=0.2$ & $\rho=0.5$ & $\rho=0.8$ \\
  \hline
     $V$ & $-0.983$ &$-1.069$ &$-1.134$ &$-1.189$ &$-1.208$ & $-1.204$\\
     $CI$ & $[-0.983,-0.982]$ & $[-0.107,-1.068]$ & $[-1.134, -0.135]$ & $[-1.190,-1.189]$ & $[-1.209,-1.207]$ & $[-1.205,-1.203]$ \\

          $V_D$ & $-1.035$ &$-1.122$ &$-1.188$ &$-1.245$ &$-1.207$ & $ -1.263$ \\
               $V_G$ & $-0.982$ &$-1.067$ &$-1.132$ &$-1.189$ &$-1.265$ & $ -1.203$ \\
  $ E_D$ &$-0.985 $ &$-1.070$ &$-1.140$ & $-1.187$ &$-1.206$ &$-1.206$\\
   Sd  &$0.411$  &$0.547$& $0.625$  & $0.626$ &$0.595$ &$0.475$\\   \hline
\end{tabular*}}
\caption{Numerical values for $V,$ $V_D$ and $V_G$ in the market situation 2 of Table \ref{parameters1} when $\gamma=0.1$ and $x_0=0$. $V$ is computed using the LMC method. CI is the confidence interval of $V$ at 99\%.   
$E_D$ is the expected utility with the strategy (\ref{new_strat}), computed using the DMC method with $10^4$ simulations per point and a time step of $T/200$ in the Euler scheme (i.e., 200 rebalances of the strategy). The last line is the standard deviation of $E_D$.}
\label{VhatVbarV}
\end{table}
 
We see that the deterministic approximations of $V$ are reliable. Moreover, looking at Table \ref{VhatVbarV}, $V_G$ is much more precise than $V_D$. Indeed, in Figure \ref{s2}, $G$ is a better approximation of $p$ than $D$. 
However, $E_D$, the expected utility computed with the strategy (\ref{new_strat}), provides an even better approximation. We have seen numerically that this is true in other market situations, which confirms the relevance of the approximated optimal strategy.

\subsection{Selecting a hedging asset}
\label{subvard}

One may examine the problem of pricing and hedging from a different point of view. Assume that an agent has the choice between different hedging assets $P$. This is for example the case for an index, where the different stocks that constitute the index may be hedging assets. She may choose the one that minimizes the reservation price with the same level of rentability and of risk i.e. $\mu$, $\sigma$ and $r$ are fixed (and thus also the Sharp ratio), and she minimizes $p$ in $\rho$. Indeed, as $p$ is a bid price, the agent may want to pay as little as possible. 
Now, we want to minimize in $\rho$ the reservation price. Unfortunately, the variation and the minimum of the reservation price (as a function of $\rho$) are very difficult to study. However, this can be done for the deterministic part $d$, as stated in the next proposition. Moreover, the numerical applications of Section \ref{secnum} suggest that $\rho^*,$ the argminimum of $d$, is in fact a good approximation of the real argminimum of the reservation price. Furthermore, as the agent select $P$ for hedging purpose, she may wish not to invest too much money. We will see numerically that choosing an hedging asset with correlation $\rho^*$ provides a very small hedging position. We finish this section with an empirical application for the determination of the risk aversion of an agent investing in the stocks composing the CAC40 index. It is worth noting that even if $g$ and $d$ appear to be similar (see \eqref{eqpetitd} and \eqref{g} in the Appendix), the study of the variations of $g$ is much more difficult. In particular, we were unable to find the minimum of $g$.
\begin{proposition}
\label{variation}
Let \begin{equation}
\label{minimum}
\rho^* = \frac{\eta T}{W\left(\l \g  s_0 \eta^2  Te^{\left(\nu  -\frac{\eta^2}{2}\right)T}  \right)}s_R.
\end{equation}
If $\rho^*\leq -1$, $d$ is increasing on $(-1,1)$. If $\rho^*\ge1$, 
 $d$ is decreasing on $(-1,1)$. \\
 Now, assume that $-1<\rho^*<1$. Then, 
$d$ is decreasing on $(-1,\rho^*)$ and increasing elsewhere and
\begin{eqnarray}
d(\rho^*) 
 =  d(0)-\frac{e^{-rT}}{2\gamma}s_R^2 T  \quad \mbox{ and }
\label{drstar1}  \quad 
g(\rho^*)
 = g(0)-\frac{e^{-rT}}{2\gamma} s_R^2 T .\nonumber
\end{eqnarray}
Moreover,
\begin{eqnarray}
\label{unifp}
\inf_{\rho\in (-1,1)}p & \geq & \left\{
    \begin{array}{ll}
        d(0)-\frac{e^{-rT}}{2\gamma}s_R^2 T  & \mbox{if } -1<\rho^*<1 \\
        \lambda e^{-rT}s_0 e^{\left(\nu-\eta s_R -\frac{\eta^2}{2}\right)T} & \mbox{if } \rho^*\geq 1 \\
        \lambda e^{-rT}s_0 e^{\left(\nu+\eta s_R -\frac{\eta^2}{2}\right)T} & \mbox{otherwise.}
    \end{array}
\right.\\
\label{unifv}
\inf_{\rho\in (-1,1)}V(x_0,\lambda)& \geq & -\frac{1}{\gamma}\exp\left(-\gamma e^{rT}\left(x_0+d(0)\right)\right) \;\;\mbox{if } -1<\rho^*<1.
\end{eqnarray}
\end{proposition}
\begin{proof}[Proof of Proportion \ref{variation}]  See Section \ref{proof_3}. $\;\square$ \end{proof}
The correlation $\rho^*$ has another interpretation related to the deterministic strategy ${\Pi}^{D,\lambda}_\rho(0,s_0)$, where ${\Pi}^{D,\lambda}_\rho$ denotes the deterministic strategy introduced in (\ref{new_strat}) for the correlation $\rho$ for all $\rho\in(-1,1)$, as we see in the following proposition.
\begin{proposition}
\label{rho_star_strat}
Assume $-1<\rho^*<1$. Then, ${\Pi}^{D,\lambda}_{\rho^*}(0,s_0)=0.$
\end{proposition}
\begin{proof}[Proof of Proportion \ref{rho_star_strat}] See Section \ref{proof_3}.  $\;\square$
\end{proof} 

We know from Section \ref{subsec_strategie} that the deterministic strategy ${\Pi}^{D,\lambda}$ is a good approximation of the optimal strategy. Thus, given an agent who buys $\lambda$ units of $S$ at time $0$ and is faced with the choice of hedging asset, Proposition \ref{rho_star_strat} shows that $\rho^*$ can be interpreted as the correlation between $P$ and $S$ for which an investor would have no share in the hedging asset $P$ at time $0$. This is relevant because the money borrowed for the hedging strategy can be very high, see Figure \ref{strategie}. We now provide similar results for $\Pi^{*,\l}_\rho$, which denotes the optimal strategy \eqref{sup_strat} for the correlation $\rho$. We first give a theoretical answer when $T$ is close to $0$ and then show numerically that choosing the correlation $\rho^*$ can provide the desired effect even when $T$ is not small.

\begin{proposition}
\label{short_mat_strat}
Let $\rho\in(-1,1)$, we have that $\lim_{T\to 0^+} {\Pi}^{*,\lambda}_\rho(0,s_0)= \frac{s_R}{\gamma\sigma}-\frac{\eta\rho}{\sigma}\lambda s_0.$ Thus,
\begin{eqnarray}
\lim_{T\to 0^+} {\Pi}^{*,\lambda}_\rho(0,s_0)=0 \iff \rho= \frac{1}{\eta s_0 \lambda \gamma}s_R=\lim_{T\to 0^+}\rho^*.
\label{short_mat_rho}
\end{eqnarray}
\end{proposition}
\begin{proof}[Proof of Proportion \ref{short_mat_strat}] See Section \ref{proof_3}. $\;\square$
\end{proof} 


Assuming that the agent holds the claims $\lambda S_T$ for a very short period of time, Proposition \ref{short_mat_strat} shows that an agent who chooses a hedging asset whose correlation is close to $s_R/(\eta s_0 \lambda \gamma)$ should not invest too much money in the hedging asset. This can also be seen numerically in the left-hand side of Figure \ref{fig_size_1}. We see numerically in the right-hand side of Figure \ref{fig_size_1} that choosing a hedging asset $P$ such that the correlation between $P$ and $S$ is equal to $\rho^*$ (or its limit in $0$) is indeed a good choice for an agent who wants to invest as little as possible on the hedging strategy, even when $T$ is not very small.

\begin{figure}[H]
  \centering
  \begin{subfigure}[b]{0.4\linewidth}
    \includegraphics[width=\linewidth]{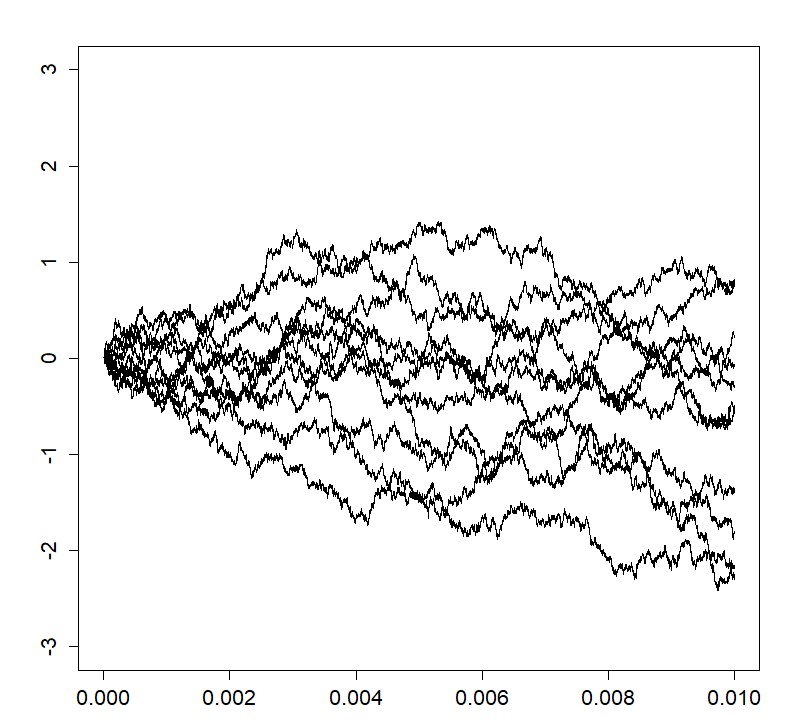}
    \caption*{Market situation 2 of Table \ref{parameters1} with $\gamma=0.1$, $T=0.01$,\\$\rho=\frac{1}{\eta s_0 \lambda \gamma}s_R=0.61875$  ($\rho^*=0.61906$).}
  \end{subfigure}
  \centering
  \begin{subfigure}[b]{0.4\linewidth}
    \includegraphics[width=\linewidth]{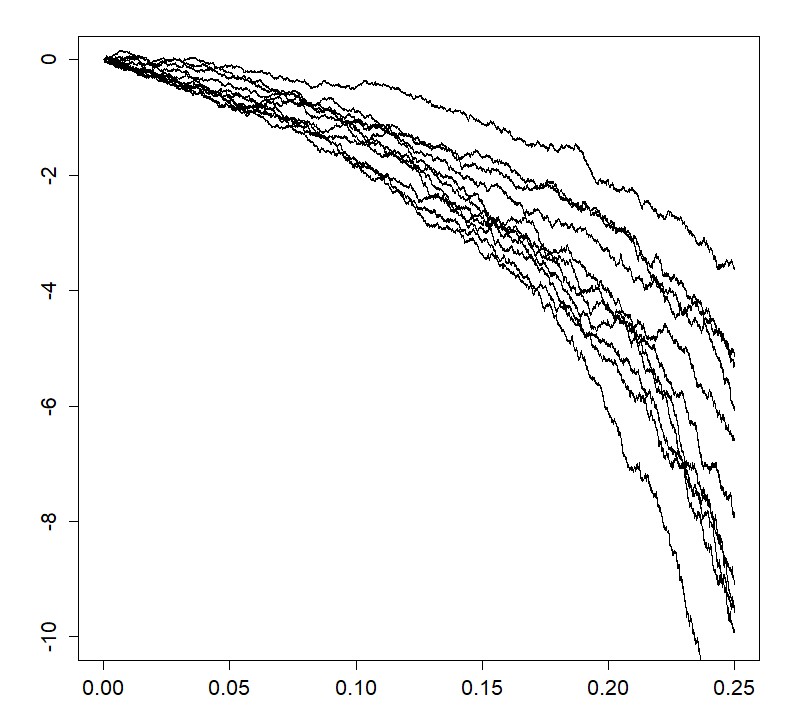}
     \caption*{Market situation 1 of Table \ref{parameters1} with $\gamma=0.5$, $T=0.25$, \\ $\rho=\rho^*=0.04$.}
  \end{subfigure}
  \caption{Several simulations of $t\mapsto \Pi_{\rho}^{D,\l}(t,S_t)$ when $\rho=\frac{1}{\eta s_0 \lambda \gamma}s_R$ or $\rho=\rho^*$.}
\label{fig_size_1}
\end{figure}


We then show numerically that even if the agent chooses to hedge herself with the asset with correlation $\rho^*$, it does not cripple the superhedging probability $\mathbb{P}(X_T^{-p(\rho^*),\Pi^{D,\lambda}}+\lambda S_T\geq 0)$ for ${\Pi}^{D,\lambda}$ defined in (\ref{new_strat}). To do that, we compute a Monte Carlo estimation (with $10^4$ simulations per point and a time step of $T/200$ in the Euler scheme, i.e., 200 rebalances of the strategy) of this probability in the situations 1 and 2 of Table \ref{parameters1} for different values of $\rho$. 

\begin{table}[H]
\resizebox{\textwidth}{!}{
\begin{tabular*}{1.13\textwidth}{|c|ccccccc|}
  \hline
   & & $\rho=-0.8$ & $\rho=-0.4$ & $\rho=0$ & $\rho=\rho^*=0.04$ & $\rho=0.4$ & $\rho=0.8$ \\
   \hline
   Situation 1 & Probability & $0.999$ &$0.999$ &$1.000$ &$0.998$ &$0.999$ & $0.999$\\
   with $\gamma=0.5$ & CI& $[0.998,1.000]$ &$[0.999,1.000]$ &$[0.999,1.000]$ &$[0.997,0.999]$ &$[0.999,1.000]$ & $[0.998,1.000]$\\
	\hline   
   Situation 2 & Probability & $0.619$ &  $0.615$ & $0.612$  &$0.600$ &$0.589$ & $0.593$\\
   with $\gamma=0.1$ &CI & $[0.606,0.632]$ &$[0.603,0.628]$ &$[0.599,0.624]$ &$[0.587,0.612]$ &$[0.576,0.602]$ & $[0.580,0.605]$\\
  \hline
\end{tabular*}}
\caption{Numerical approximations of the superhedging probability}
\end{table}

This suggest that the superhedging probability does not change much when varying the correlation and justifies the choice of a hedging asset $P$ with correlation $\rho^*$.\\
%
%
%
%

We finish with an approximation of $\rho^*$ for small and large values of $T$.
\begin{lemma}
\label{lim_rho_star}
We determine that 
\begin{eqnarray}
\rho^* & = &s_R\left(\frac{1}{\eta s_0 \lambda \gamma}+\left(\eta  - \frac{\nu-\frac{\eta^2}{2}}{\eta s_0 \lambda\gamma}\right)T\right)+T\epsilon(T),
\label{t_app}
\end{eqnarray}
where $\lim_{T\to 0^+}\epsilon(T)=0$.
\label{closed}
Assume that $\nu > \eta^2/2.$ Then, $\lim_{T\to +\infty}\rho^* = \frac{\eta}{\nu-\frac{\eta^2}{2}}s_R$. 
\end{lemma}
\begin{proof}[Proof of Lemma \ref{closed}]  See Section \ref{proof_3}. $\;\square$ \end{proof}
Table \ref{tab_relative_error} shows the efficiency of the first-order approximation in (\ref{t_app}).
\begin{table}[H]
\begin{tabular*}{\textwidth}{|l@{\extracolsep{\fill}}|ccccc|}
  \hline
    & $T=0.01$ & $T=0.1$ & $T=0.5$ & $T=0.8$ & $T=1$  \\
  \hline
    $\rho^*$ & $0.310$ &$0.321$ &$0.357$ &$0.379$ &$0.391$\\
   Relative error & $<0.01\%$ &  $0.15\%$ & $2.76\%$  &$5.87\%$ &$8.23\%$  \\
   \hline
\end{tabular*}
\caption{Relative error between $\rho^*$ and the first-order approximation in (\ref{t_app}) in the market situation 2 of Table \ref{parameters1} when $\gamma=0.2$.}
\label{tab_relative_error}
\end{table}

We may also use \eqref{minimum} in order to estimate the risk aversion $\gamma$ of the investor. Consider the CAC40 index to be the non-traded asset and assume that the agent invests on some stocks of the CAC40 index. The calibration of the data to the model \eqref{mdl} for each stock composing the CAC40 index was carried out  with the daily closing prices of the stocks starting from date 2021-01-01 until date 2022-01-01 and with $r=-0.005$. We have used the K-means algorithm in order to identify a cluster of 6 stocks having a similar Sharp ratio (and belonging to a cluster with centroid $\hat{s}_R=2.26$). Those stocks belong to the following companies : Teleperformance, Hermès, Axa, Pernod Ricard, Veolia and Dassault Systèmes. We have also estimated the parameters of the CAC40 index, $\nu=0.14$ and $\eta=0.26$. Table \ref{table_gamma_sr} gives the estimated correlation $\hat{\rho}$ of those stocks with the CAC40 index as well as their estimated Sharp ratios $\hat{s}_R$. If we assume that the agent wants to cover a long position on the CAC40 index and that to do so, select optimally her hedging asset, then inverting \eqref{minimum} will provide the value of her risk aversion (see Table \ref{tablegamma}). 

\begin{table}[H]
\centering
\resizebox{1\textwidth}{!}{
\begin{tabular*}{1.1\textwidth}{|l@{\extracolsep{\fill}}|cccccc|}
  \hline
    & Teleperformance & Hermès
 & Axa & Pernod Ricard & Veolia & Dassault Systèmes \\
  \hline
     $\hat{\rho}$ & $0.27$ & $0.62$ & $0.72$ & $0.55$ & $0.43$ & $0.34$ \\
     $\hat{s}_R$ & $1.79$ & $2.47$ & $1.97$ & $1.94$ & $2.19$ & $2.14$\\
     \hline 
\end{tabular*}}
\caption{For each stock, estimated Sharp ratio $\hat{s}_R$ and estimated values for its correlation $\hat{\rho}$ with the CAC40 index.}
\label{table_gamma_sr}
\end{table}

\begin{table}[H]
\centering
\resizebox{1\textwidth}{!}{
\begin{tabular*}{1.1\textwidth}{|l@{\extracolsep{\fill}}|cccccc|}
  \hline
    & Teleperformance & Hermès
 & Axa & Pernod Ricard & Veolia & Dassault Systèmes \\
  \hline
     $\gamma$ & $0.020$ & $0.005$ & $0.004$ & $0.006$  & $0.009$ & $0.013$\\
     \hline 
\end{tabular*}}
\caption{For each stock, empirical values for the risk aversion $\gamma$ that satisfy $\rho^*(\gamma)=\hat{\rho}$ with $s_0=5589$ and $s_R=\hat{s}_R$.}
\label{tablegamma}
\end{table}

\section{Short put position}
\label{secshorput}
We now study the case of a short position on a put option with payoff $(K-S_T)_+$, where $K\in(0,\infty)$. 
Recall that $V^{put}(x_0,\l,K)$ and $p^{put}$ are respectively the value function of a short position on $\l$ units of a put option and its ask reservation price, i.e. see \eqref{indiffsell_gen}, $$V^{put}(x_0,\l,K)=V(x_0,\l,-(K-x)_+)\;\; \mbox{and} \;\;  p^{put}=p^{sell}_{(K-x)_+}.$$  
\begin{theorem}
\label{prop_decput}
Let $\bar w=W\left(s_0 e^{(\nu-\eta\rho s_R -\frac{\eta^2}{2})T} \eta^2 T \lambda\gamma(1-\rho^2) \right)$ and suppose that
\begin{eqnarray}
K\geq \frac{1}{\lambda\gamma(1-\rho^2)}\left(\frac{\bar w}{\eta^2 T}+\frac{\bar w^2}{2\eta^2 T}\right).
\label{condition_put}
\end{eqnarray} 
Then, $p^{put}=D^{put}+A^{put},$
where
\begin{eqnarray*}
D^{put} & = & \lambda e^{-rT}K -\frac{e^{-rT}}{\gamma(1-\rho^2)}\left(\frac{\bar w}{\eta^2 T}+\frac{\bar w^2}{2\eta^2 T} \right)  \mbox{ and }
A^{put} =\frac{e^{-rT}}{\gamma(1-\rho^2)}\ln \mathbb{E}\left(\psi_K(N)\right) \\
\psi_K(y) & = & \exp\left[- \frac{\bar w}{\eta^{2}T}\left(e^{\eta\sqrt{T} y}-1-\eta\sqrt{T} y\right)+\left(\frac{\bar w}{\eta^2T}e^{\eta\sqrt{T} y}-\l\g(1-\r^2) K\right)_+\right].
\end{eqnarray*}
Moreover, $V^{put}(x_0,\l,K)=V_D^{put}(x_0,\lambda,K)V_A^{put}(\lambda,K),$
where
\begin{eqnarray*}
V_D^{put}(x_0,\lambda,K) &=& -\frac{1}{\gamma}\exp\left[{-\gamma e^{rT} \left(x_0-D^{put}\right)-\frac{s_R^2}{2}T}\right]\\ 
V_A^{put}(\lambda,K) &=&  \exp \left[\gamma  e^{rT} A^{put}\right].
\end{eqnarray*}
\end{theorem}  
\begin{proof}[Proof]
Recall that $\hat{s}_0=s_0e^{(\nu-\eta\rho s_R-\frac{\eta^2}{2})T}$ and $\hat{\theta}=\l\g(1-\r^2)$. As $p^{put}=-p_{-(K-x)_+}$, we want to use Theorem \ref{prop_dec_gen} with $K<+\infty$, $\zeta(x)=x-K,$  $v=\hat{s}_0$ and $u=-vK/\hat{s}_0=-K$. Then, $\hat{\zeta}(x)=\hat{s}_0x-K$. 
As (\ref{condition}) is equivalent to 
$\frac{\bar w}{\eta^{2}T}
+ \frac{\bar w^2}{2\eta^{2} T}\leq {K}\l\g(1-\r^2)$, it is true by \eqref{condition_put}. Thus, it remains to prove that 
$I_{\hat{\zeta},\hat{K}}(\hat{\theta}) = \mathbb{E}\left(\psi_K(N)\right).$ This follows from 
\begin{eqnarray*}
&&\phi_{\hat{K}}(y,\hat{\theta})  =\exp\left(- \frac{\bar w}{\eta^2 T}\left(e^{\eta\sqrt{T} y}-1-\eta\sqrt{T} y\right)+\left(\frac{\bar w}{\eta^2 T}e^{\eta\sqrt{T}y}-\hat{\theta} K\right)_+ \right)\\
&&\hat{\zeta}\left(\frac{W(\hat{\theta} v \eta^2 T)}{\hat{\theta} v\eta^2 T} e^{\eta\sqrt{T}N}\right)-u-v \frac{W(\hat{\theta} v \eta^2 T)}{\hat{\theta} v\eta^2 T}e^{\eta\sqrt{T}N}=0.
\end{eqnarray*}
For the decomposition of the value function, \eqref{indiffsell_gen} and \eqref{value_pe} imply that 
\begin{eqnarray*}
V^{put}(x_0,\lambda,K)&=&-\frac{1}{\gamma}\exp\left[-\gamma e^{rT}\left(x_0-p^{put}\right)-\frac{s_R^2}{2}T\right]. \quad\square
\end{eqnarray*}
\end{proof}
Condition \eqref{condition_put} is essentially asking that the option is in the money. Indeed, \eqref{condition_put} is equivalent to $Ke^{-rT}\geq \frac{D}{\lambda}$ where $D$ is defined in \eqref{eqD1} and $\frac{D}{\lambda}$ is a reliable approximation of the per unit reservation price of a long position of $\lambda$ units of the non-traded asset $S$. We provide in Table \ref{value_p_borne}, the minimal value for the strike $K$. We see that this minimal value depends on the market situation. 

\begin{table}[H]
\begin{tabular*}{1\textwidth}{|l@{\extracolsep{\fill}}|ccccc|}
  \hline
    & $\rho=-0.5$ & $\rho=-0.25$ & $\rho=0$ & $\rho=0.25$ & $\rho=0.5$ \\
  \hline
   Situation 1 & $66.55$ &$61.86$ &$60.23$ &$61.10$ &$64.80$ \\
   Situation 2 & $0.95$ &$0.91$ &$0.89$ &$0.89$ &$0.90$ \\
   Situation 3 & $9.67$ &$7.49$ &$6.57$ &$6.33$ &$6.84$\\
   \hline
\end{tabular*}
\caption{Values of $\frac{1}{\lambda\gamma(1-\rho^2)}\left(\frac{\bar w}{\eta^2 T}+\frac{\bar w^2}{2\eta^2 T}\right)$ in the market situations of Table \ref{parameters1} for several values of $\rho$ when $\gamma=0.5$.}   
\label{value_p_borne}
\end{table}

We provide in Figure \ref{lambert_direct_put} numerical simulations for the reservation price $p^{put}$. 
\begin{figure}[h]
\captionsetup[subfigure]{justification=centering}
  \centering
  \begin{subfigure}[b]{0.4\linewidth}
    \includegraphics[width=\linewidth]{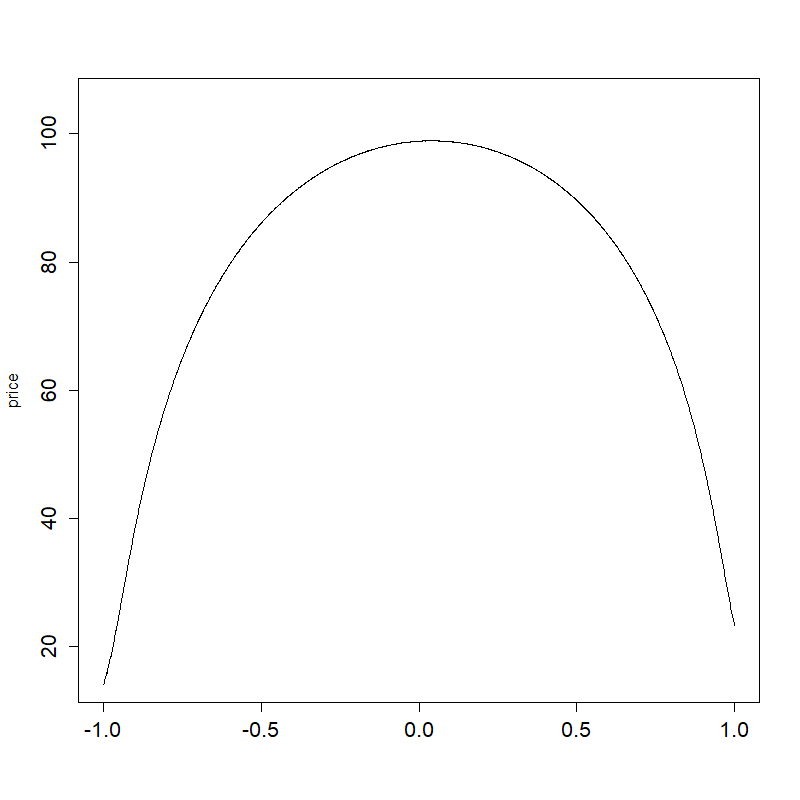}
     \caption*{Lambert Monte Carlo method}
  \end{subfigure}
  \begin{subfigure}[b]{0.4\linewidth}
    \includegraphics[width=\linewidth]{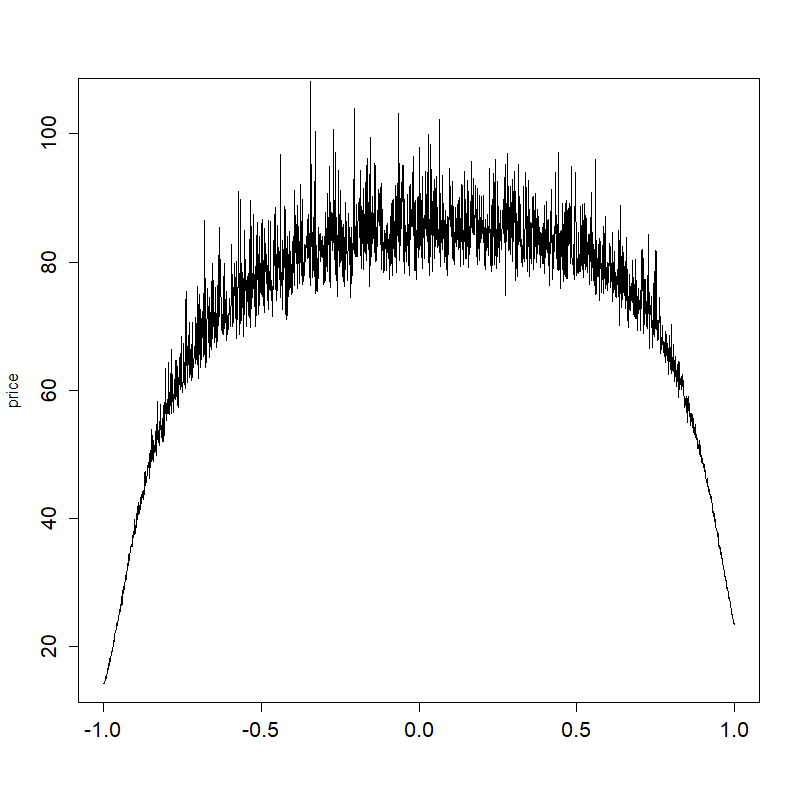}
    \caption*{Direct Monte Carlo method}
  \end{subfigure}
  \caption{Computation of the ask reservation price $p^{put}$ as a function of  $\rho$ in the market situation 1 of Table \ref{parameters1} with $\gamma=0.5$ and $K=110$ with $10^4$ simulations per point.}
  \label{lambert_direct_put}
\end{figure}

As in the long stock position case, the DMC estimator provides an estimator with high variance. Conversely, the LMC estimator exhibits low variance. For example, in the market situation 1 of Table \ref{parameters1}, with $\gamma=0.5$ and $\rho=0$, the variance is 4.745 for the DMC method (for $10^4$ simulations per point) while the variance of the LMC method is 0.007.  
The results are similar in the market situation 2 of Table \ref{parameters1} with $\gamma=5$ and $K=s_0$ with $10^4$ simulations per point.

We now propose in Table \ref{table_m1_put},  numerical simulations for the value function. In the market situation 1 of Table \ref{parameters1}, when $x_0=75$, $\gamma=0.5$ and $K=s_0=100$, $V^{put}$ and $V_D^{put}$ are large enough to be relevant for correlation $\rho$ such that $|\rho|\leq 0.5$. Indeed, for higher values of $\rho$, the values of $V^{put}$ and $V_D^{put}$ are very small and the numerical estimations are not relevant (for example, if $\rho=0.7$, $V\simeq -1.9.10^{-4}$ and $V^{put}_D=-1.8.10^{-4}$). 
\begin{table}[H]

\begin{tabular*}{1\textwidth}{|l@{\extracolsep{\fill}}|ccccc|}
  \hline
    & $\rho=-0.5$ & $\rho=-0.25$ & $\rho=0$ & $\rho=0.25$ & $\rho=0.5$ \\
  \hline
   $V^{put}$ & $-0.023$ &$-2.611$ &$-13.336$ &$-5.718$ &$-0.130$ \\
   $\mbox{CI}$ & $[-0.023,-0.022]$ &$[-2.585, -2.367]$ &$[-13.466,-13.208]$ &$[-5.777,-5.661]$ &$[-0.193, -0.132]$ \\
   $V_D^{put}$ & $-0.034$ &  $-3.654$ & $-18.531$  &$-8.035$ &$-0.193$\\
   \hline
\end{tabular*}
\caption{$V^{put}$ and $V_D^{put}$ in the market situation 1 of Table \ref{parameters1} when $\gamma=0.5$ and $x_0=75$. CI is a 99\% confidence interval of $V^{put}$. }
\label{table_m1_put}
\end{table}
We see that the deterministic approximation of $V^{put}$ is still reliable even if it is less faithful than the one for the long stock position. 
 Note that in the market situation 2 of Table \ref{parameters1} with $\gamma=0.1$ and $x_0=0$, (\ref{condition_put}) does not hold true.\\ 

\section{Conclusion and discussion}
\label{conclu}
We have presented a methodology to obtain an improved Monte Carlo  estimate of the reservation price of a derivative in a Black-Scholes market. 
The reservation price is a utility based price and is useful in an incomplete market where the underlying asset cannot be traded dynamically, hence the pricing by the martingale approach is not feasible. An alternative asset which can be traded dynamically in the same market is then used to hedge the position of the investor. The (bid) reservation price  is the amount which leaves the  trader indifferent to have the derivative or not. The reservation price can be written as the logarithm of the Laplace transform of a function of a log-normal random variable, that we have decomposed in two parts using the Lambert function. The first one is deterministic and the second one can be seen as an expected value which is closed to one. 
This decomposition provides an improved numerical method called the "LMC method". In the case of a long position on the non-traded asset, it allows the construction of upper and lower deterministic bounds for the reservation price and the value function, as well as an approximation of the optimal strategy and of some Greeks. The lower bound is a reliable approximation of the (bid) reservation price, which can be studied analytically and we obtained information about its behavior as well as its argminimum $\rho^*$. We have showed that this minimum may be the optimal choice of correlation for an agent who want to select an hedging asset $P$ that minimizes the reservation price and also the cash invested in $P$ retaining the same level of rentability.  We have shown numerically that choosing such a $P$ does not cripple the superhedging probability.
We were able to control the error between the bounds and the reservation price. Results for a short position on a put option have also been provided. One may ask if the same conclusions apply to the case of a long position on a call option. 
The decomposition of the bid reservation price of a long position on a call option can be obtained using similar methods (the details are not provided in the paper for sake of brevity). 
We introduce an adapted notion of  in- and out-of-the-money call options. An in-the-money (resp. out-of- or at-) call option of strike $K$ is an option, where $K>K^*$ (resp. $K<K_*$ or $K_* \leq K\leq K^*$), where 
$$K_*=\frac{\bar{w}}{\l\g(1-\r^2)\eta^2 T}\quad \mbox{ and } \quad K^*=s_0e^{\left(\nu-\eta\rho s_R-\frac{\eta^2}{2}\right)T},$$
see Theorem \ref{prop_dec} for the definition of $\bar{w}$. We easily see that $K_*\leq K^*.$
The LMC method provides better results for  in- and out-of-the-money call options, but for at-the-money call options, it is not clearly better than when using the DMC method.  Moreover, if $K< K_*$ or $K> K_*,$ the deterministic part of the decomposition of the reservation price is a good approximation of the price. This is illustrated in Figure \ref{scll} below.\\ 
\begin{figure}[H]
\captionsetup[subfigure]{justification=centering}
  \centering
  \centering
\resizebox{0.66\textwidth}{!}{
  \begin{subfigure}[b]{0.37\linewidth}
    \includegraphics[width=\linewidth]{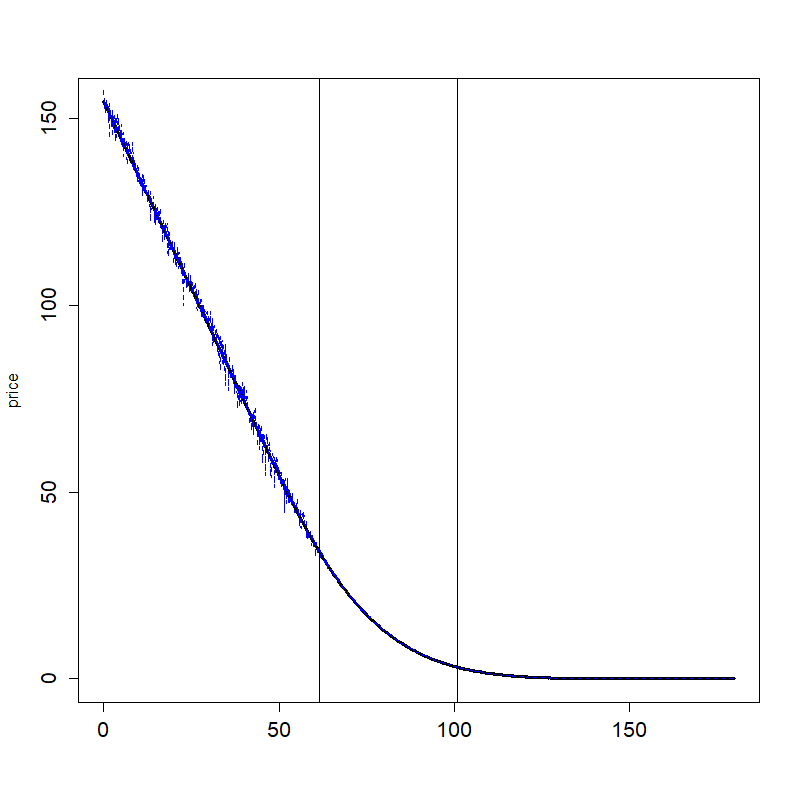}
    \caption*{LMC method and DMC method}
  \end{subfigure}
  \centering
  \begin{subfigure}[b]{0.37\linewidth}
    \includegraphics[width=\linewidth]{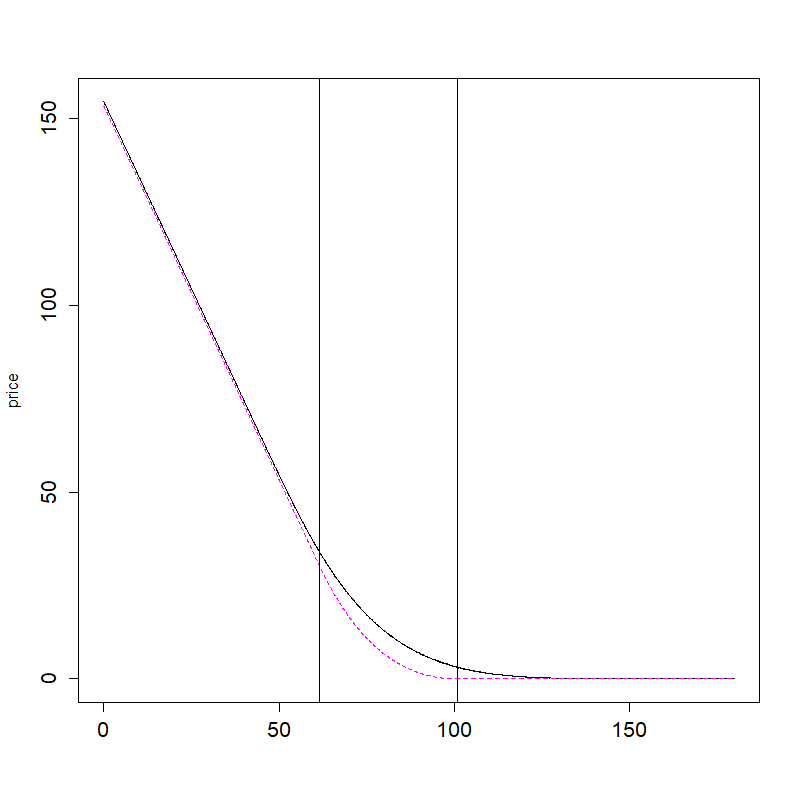}
     \caption*{Deterministic approximation}
  \end{subfigure}}
  \caption{In the market situation 1 of Table \ref{parameters1} with $\gamma=0.5$ and $\rho=0.8$, estimated bid reservation price of $\l (S_T-K)_+$ as a function of $K$,  LMC method in black, DMC method in blue (dashed line), and deterministic approximation in purple (dashed line). The price is simulated with $10^4$ simulations per point. The first vertical line is $K_*$ and the second one is $K^*$.
}
  \label{scll}
\end{figure}
One may ask to what extend one can go beyond the Black-Scholes model. One possible generalization would be to have several tradable assets to cover $S$.  To the best of our knowledge, no semi-closed formula like  \eqref{indiffdef_gen} has been derived in this case. Another generalization could be to consider that the parameters of $S$ and $P$ are time dependent. In this case, it is still possible to obtain a formula like \eqref{indiffdef_gen} and to perform our method.
Now if the coefficients are also stochastic, Tehranchi showed in \citep[Remark 5]{ref14} that the reservation price can always be written in a semi-closed formula though the resulting expression can be very different from \eqref{indiffdef_gen} and does not necesseraly imply the Laplace transform of a function of a lognormal random variable. One last possible generalization would be to consider a random drift coefficient $\nu$ (at time $t=0$) with a gaussian prior. This is useful in the context of drift erroneous estimate. In that situation, Monoyios provides in \cite{ref4} a semi-closed expression for the reservation price depending on the minimal martingale measure but still very similar to (\ref{indiffdef_gen}). Moreover, one can easily show that the value of the non-traded asset at time $T$ still follow a lognormal distribution under the minimal martingale measure. This means that our method would also be appropriate in this case.

\section{Appendix}
\label{secappendix}
\subsection{Taylor approximation of the reservation price}
\label{apptaylor}
As announced in Section \ref{taylorapprox}, we provide the Taylor expansion of the reservation price when the correlation is close to $1$. Let  $G =  e^{\eta\sqrt{T}N-\frac{\eta^2 T}{2}}$ and $g_c$ be the cumulant-generating function of $G,$ i.e., $g_c(t)=\ln(\mathbb{E}(e^{tG}))$. Since $G$ is a log-normal random variable, $g_c$  exists only for negative values of $t$.
Furthermore, as $G$ has moments of any orders, the function $g_c$ is $C^\infty$ on $(-\infty,0),$ with limits at $0^-,$ for all derivatives.
So, the function $g_c$ admits a Taylor expansion up to an arbitrary order $n>0$ at $0^-$ and has the following decomposition :
\begin{eqnarray}
\forall t\in (-\infty,0) \quad g_c(t)=\sum_{k=1}^{n}{\chi_k\frac{t^k}{k!}}+t^n\epsilon(t)\quad \mbox{with} \quad \underset{t\rightarrow0^-}{\lim}\epsilon(t)=0.\label{dl}
\end{eqnarray}
The coefficients $\chi_k$ are called the cumulants of the distribution of $G$ and have the following expressions :
\begin{eqnarray*}
\chi_1 & = & 1 \quad \chi_2  =  e^{\eta^2T}-1\\
\chi_3 & = &e^{3\eta^2T}-3e^{\eta^2T}+2 \quad
\chi_4  = e^{6\eta^2T}-4e^{3\eta^2T}-3e^{2\eta^2T}+12e^{\eta^2T}-6\\
\chi_5 & = &e^{10\eta^2T}-5e^{6\eta^2T}-10e^{4\eta^2T}+20e^{3\eta^2T}+30e^{2\eta^2T}-60e^{\eta^2T}+24.
\end{eqnarray*}
As a function of $\rho$, the reservation price can be written as follows (see, (\ref{indiffdef_gen}) with $h=id$):
\begin{eqnarray}
p(\rho)&=&{ -\frac{e^{-rT}}{\gamma(1-\rho^2)}g_c\left(-\lambda\gamma (1-\rho^2) s_0e^{\left(\nu-\eta\rho s_R \right) T}\right)} = \sum_{k=0}^{n}c_k(1-\rho)^k+(1-\rho)^n \e(\rho),
\label{taylor}
\end{eqnarray} 
with $\lim_{\rho\rightarrow1^-}\e(\rho)=0$ and where the coefficients $(c_k)_{k>0}$ can be expressed in terms of the sequence $(\chi_k)_{k>0},$ for any arbitrary order $n\geq 0$. 
Set $\alpha  =  \eta s_R T$, $\o=1-\rho$, $\overline{p}=\lambda s_0 e^{(\nu-\eta s_R)T}$ and $\hat{p}=e^{-rT}\overline{p}$. Note that $\hat{p}$ corresponds to the expectation of $\lambda e^{-rT}S_T$ computed under the unique probability measure which turns the discounted value of the traded asset $P$ into a martingale in the complete case, i.e., when $\rho=1$. Indeed,
\begin{eqnarray*}
\frac{d\widetilde{S}_t}{\widetilde{S}_t}= \left(\nu-r-\eta s_R\right)dt+\frac{\eta}{\sigma}\frac{d\widetilde{P}_t}{\widetilde{P}_t}.
\end{eqnarray*}
We also introduce $\f(\o)= \overline{p}\gamma \o (\o-2) e^{\alpha \o}.$ 
The reservation price $p$ takes the following form as a function of $\o$
\begin{eqnarray}
p(\o)= 
-\frac{e^{-rT}}{\g \o(2-\o)}\ln\left(\mathbb{E}\left(e^{\f(\o)G}\right)\right).
\label{new_writting_pe}
\end{eqnarray}
We first compute  the Taylor expansion of order $5$   at $0^+$ of  $\ln (\mathbb{E}(e^{\f G}))$. 
For that we use \eqref{dl} for $n=5$ and as $\lim_{\o \to 0^+}\f(\o)/ \o =0,$ we obtain that 
\begin{multline}
\ln \left(\mathbb{E}\left(e^{\f(\o)G}\right)\right) = \chi_1 \f(\o)+\frac{\chi_2}{2}\f^2(\o)+\frac{\chi_3}{6}\f^3(\o)+  \frac{\chi_4}{24}\f^4(\o)+\frac{\chi_5}{120}\f^5(\o) + \o^5 \varepsilon(\o),
\label{taylor_cumulant}
\end{multline} 
where $\varepsilon$ will denote in the sequel a function such that  $\lim_{\o \to 0^+} \varepsilon(\o)=0$. Now, we compute the Taylor expansion of order $5$ of the powers of $\f$. 
\begin{align*}
\f(\o) 
&= \overline{p}\gamma (\o^2-2\o) \left(1+\alpha \o + \frac{\alpha^2}{2}\o^2+\frac{\alpha^3}{6}\o^3+\frac{\alpha^4}{24}\o^4+\o^4 \varepsilon(\o)\right)\\
&= \overline{p}\g \left(-2\o + (1-2\alpha)\o^2+\alpha(1-\alpha)\o^3+ \alpha^2\left(\frac{1}{2}-\frac{\alpha}{3}\right)\o^4+ \frac{\alpha^3}{6}\left(1-\frac{\alpha}{2}\right)\o^5 \right)+\o^5 \varepsilon(\o))\\
\f^2(\o) &=  (\overline{p}\g)^2 \left(4\o^2+4(2\alpha-1)\o^3+ \left(8\alpha^2-8\alpha+1\right)\o^4+ \alpha\left(\frac{16}{3}\alpha^2-8\alpha+2\right)\o^5 \right)+\o^5 \varepsilon(\o)\\
\f^3(\o) &= \begin{multlined}[t] (\overline{p}\g)^3 \left(-8\o^3+12\left(1-2\alpha\right)\o^4+6\left(-6\alpha^2+6\alpha-1\right)\o^5 \right)+\o^5 \varepsilon(\o)
\end{multlined}\\
\f^4(\o) &= \begin{multlined}[t] (\overline{p}\g)^4 \left(16\o^4+32\left(2\alpha-1\right)\o^5 \right)+\o^5 \varepsilon(\o)\end{multlined}\\
\f^5(\o) &= \begin{multlined}[t] -32 (\overline{p}\g)^5\o^5+\o^5 \varepsilon(\o).\end{multlined}
\end{align*}
Plugging the previous expressions in (\ref{taylor_cumulant}), we find that 
$
\ln (\mathbb{E}(e^{\f(\o)G})) = \sum_{i=1}^5 a_i\o^i + \o^5 \varepsilon(\o),
$
where,
\begin{align*}
a_1 &= -2\chi_1  \overline{p}\gamma \quad a_2 = \chi_1 (1-2\alpha)\overline{p}\gamma+2\chi_2 (\overline{p}\gamma)^2\\
a_3 &= \chi_1\alpha(1-\alpha)\overline{p}\gamma+2\chi_2 (2\alpha-1)(\overline{p}\gamma)^2-  \frac{4}{3}\chi_3(\overline{p}\gamma)^3\\
a_4 &= \frac{1}{6}\chi_1\alpha^2(3-2\alpha)\overline{p}\gamma+\frac{1}{2}\chi_2 (8\alpha^2-8\alpha+1)(\overline{p}\gamma)^2+ 2\chi_3 (1-2\alpha)(\overline{p}\gamma)^3+\frac{2}{3}\chi_4  (\overline{p}\gamma)^4\\
a_5 &= 
\frac{1}{12}\chi_1\alpha^3(2-\alpha)\overline{p}\gamma+\frac{1}{3}\chi_2 \alpha(8\alpha^2-12\alpha+3)(\overline{p}\gamma)^2+ \chi_3  (-6\alpha^2+6\alpha-1)(\overline{p}\gamma)^3+\frac{4}{3}\chi_4  (2\alpha-1)(\overline{p}\gamma)^4 -\frac{4}{15}\chi_5(\overline{p}\gamma)^5. 
\end{align*}
We now proceed to the expansion of $p$ in \eqref{new_writting_pe}. 
\begin{align*}
p(\o)&=-\frac{e^{-rT}}{2\gamma}\frac{1}{\left(1-\frac{\o}{2}\right) \o}\ln\left(\mathbb{E}\left(e^{\f(\o)G}\right)\right) = -\frac{e^{-rT}}{2\gamma}\left(1+\frac{\o}{2}+\frac{\o^2}{4}+\frac{\o^3}{8}+\frac{\o^4}{16}+\o^4\varepsilon(\o)\right) \frac{1}{\o} \left( \sum_{i=1}^5 a_i \o^i+\o^5 \varepsilon(\o) \right) \\
& = \begin{multlined}[t]
-\frac{e^{-rT}}{2\gamma} \bigg(a_1+\left(\frac{a_1}2+a_2\right) \o^1+ \left(\frac{a_1}4+\frac{a_2}2+a_3\right) \o^2 
+ \left(\frac{a_1}8+\frac{a_2}4+ 
 \frac{a_3}2+a_4\right) \o^3  \\ +  \left(\frac{a_1}{16}+\frac{a_2}8+\frac{a_3}4+\frac{a_4}2+a_5\right) \o^4 \bigg) +\o^4 \varepsilon(\o). \end{multlined}
\end{align*}
Recalling that $\o=1-\rho$ and noting that for all $1 \leq k\leq 4$ , $-2\gamma e^{r T} c_{k}=a_{k+1}-\gamma e^{r T}c_{k-1}$,  we determine that
\begin{align*}
-2\gamma e^{rT}c_0 &= -2  \overline{p}\gamma\\
-2\gamma e^{rT}c_1 &=  (1-2\alpha)\overline{p}\gamma+2\chi_2 (\overline{p}\gamma)^2 -   \overline{p}\gamma = -2   \alpha\overline{p}\gamma +2\chi_2(\overline{p}\gamma)^2  \\
-2\gamma e^{rT}c_2 &= \alpha(1-\alpha)\overline{p}\gamma+2\chi_2 (2\alpha-1)(\overline{p}\gamma)^2-  \frac{4}{3}\chi_3(\overline{p}\gamma)^3 + \chi_2(\overline{p}\gamma)^2 - \alpha \overline{p}\gamma  = - \alpha^2 \overline{p}\gamma + \chi_2  (4\alpha-1)(\overline{p}\gamma)^2-\frac{4}{3}\chi_3(\overline{p}\gamma)^3\\
-2\gamma e^{rT}c_3 &= \frac{1}{6}\alpha^2(3-2\alpha)\overline{p}\gamma+\frac{1}{2}\chi_2 (8\alpha^2-8\alpha+1)(\overline{p}\gamma)^2+ 2\chi_3 (1-2\alpha)(\overline{p}\gamma)^3+\frac{2}{3}\chi_4(\overline{p}\gamma)^4  -\frac{1}{2}\alpha^2\overline{p}\gamma  \\ &\;\;\;\; + \frac{1}{2}\chi_2  (4\alpha-1)(\overline{p}\gamma)^2-\frac{2}{3}\chi_3(\overline{p}\gamma)^3 = -\frac{1}{3}\alpha^3\overline{p}\gamma+2\chi_2  \alpha(2\alpha-1)(\overline{p}\gamma)^2 +\frac{4}{3}\chi_3 (1-3\alpha)(\overline{p}\gamma)^3+\frac{2}{3}\chi_4(\overline{p}\gamma)^4\\
-2\gamma e^{rT}c_4 &=  \frac{1}{12}\alpha^3(2-\alpha)\overline{p}\gamma+\frac{1}{3}\chi_2 \alpha(8\alpha^2-12\alpha+3)(\overline{p}\gamma)^2+ \chi_3 (-6\alpha^2+6\alpha-1)(\overline{p}\gamma)^3+\frac{4}{3}\chi_4 (2\alpha-1)(\overline{p}\gamma)^4 -\frac{4}{15}\chi_5(\overline{p}\gamma)^5 \\
&\;\;\;\; -\frac{1}{6}\alpha^3\overline{p}\gamma+\chi_2  \alpha(2\alpha-1)(\overline{p}\gamma)^2 + \frac{2}{3}\chi_3 (1-3\alpha)(\overline{p}\gamma)^3+\frac{1}{3}\chi_4(\overline{p}\gamma)^4\\ &= -\frac{1}{12}\alpha^4\overline{p}\gamma + \frac{2}{3}\chi_2 \alpha^2 (4\alpha-3)(\overline{p}\gamma)^2+\chi_3 \left(-6\alpha^2+4\alpha-\frac{1}{3}\right) (\overline{p}\gamma)^3+\frac{1}{3}\chi_4(8\alpha-3)(\overline{p}\gamma)^4-\frac{4}{15}\chi_5 (\overline{p}\gamma)^5.
\end{align*}
So, we obtain that 
\begin{eqnarray*}
c_0& = &\hat{p}\quad
c_1 = \alpha\hat{p}-\chi_2\gamma e^{rT}\hat{p}^2 \quad
c_2 = \frac{1}{2}\alpha^2\hat{p}-\frac{1}{2}\chi_2\gamma(4\alpha-1)e^{rT}\hat{p}^2+\frac{2}{3}\chi_3\gamma^2e^{2rT}\hat{p}^3\\
c_3& = &\frac{1}{6}\alpha^3\hat{p}-\chi_2\gamma\alpha(2\alpha-1)e^{rT}\hat{p}^2+\frac{2}{3}\chi_3\gamma^2(3\alpha-1)e^{2rT}\hat{p}^3-\frac{1}{3}\chi_4\gamma^3 e^{3rT}\hat{p}^4\\
c_4& = & \frac{1}{24}\alpha^4\hat{p}-\frac{1}{3}\chi_2\gamma\alpha^2(4\alpha-3)e^{rT}\hat{p}^2+\chi_3\gamma^2\left(3\alpha^2-2\alpha+\frac{1}{6}\right)e^{2rT}\hat{p}^3-\frac{1}{6}\chi_4\gamma^3(8\alpha-3)e^{3rT}\hat{p}^4+\frac{2}{15}\chi_5\gamma^4e^{4rT}\hat{p}^5.
\end{eqnarray*}

\subsection{Proofs of Section \ref{secdecomp}}
\label{proof_1}
\subsubsection{Properties of the Lambert function}
In this section, we provides some useful properties of the Lambert function that will be used in the proofs of our results. By definition, $W$ has the following properties: 
\begin{eqnarray}
\label{eqlamb2}
\forall x>-\frac{1}{e} \quad  W(x)e^{W(x)}  =  x   \mbox{ and }
\forall x>-1\quad W(xe^{x})=x.
\end{eqnarray}
\begin{lemma}
\label{lemlamb}
The Lambert function $W$ is strictly increasing $C^{\infty}$ and \begin{equation}
\label{eqlambderiv}
\forall x>-\frac{1}{e},\; W'(x)=\frac{1}{e^{W(x)}+x}=\frac{W(x)}{x(1+W(x))}.
\end{equation}
Moreover, $W$ admits a Taylor expansion around $0$ : $W(x)=x-x^{2}+x^2\epsilon(x) \mbox{ with } \underset{x \rightarrow 0}{\lim}\; \epsilon(x)=0.$\\
Finally, $W(x)\underset{x \rightarrow +\infty}{\sim}\ln(x)$.
\end{lemma}
Note that Lemma \ref{lemlamb} and \eqref{eqlamb2} imply that  $
W(0)=0$ and thus that $W(x)>0, \quad \forall x>0. $
\begin{proof}[{Proof of Lemma \ref{lemlamb}}]
First, $W$ is clearly $C^{\infty}$, being the inverse of $x\mapsto x e^x$. Now, we differentiate \eqref{eqlamb2} and obtain that for $x>-\frac{1}{e}$,  
\begin{eqnarray}
\nonumber
(We^W)'(x)=1 & \iff & W'(x)e^{W(x)}+W'(x)W(x)e^{W(x)}=1  \\
\nonumber
& \iff & W'(x)e^{W(x)}\left(1+W(x)\right)=1 \\ 
\label{eqlamb3}
& \iff & W'(x)=\frac{1}{e^{W(x)}+x}=\frac{W(x)}{x(1+W(x))},
\end{eqnarray}
where we have used for the last equivalence the fact that $\forall x>-\frac{1}{e},\;W(x)>-1$ and \eqref{eqlamb2}.\\
As $\forall x>-\frac{1}{e},\; W'(x)>0$, $W$ is strictly increasing. Recalling that $W$ is $C^{\infty}$, we determine that 
$$W(x)=W(0)+W'(0)x+\frac{W''(0)}{2}x^2+x^2\epsilon(x)\quad \mbox{ with }\quad \underset{x \rightarrow 0}{\lim}\; \epsilon(x)=0.$$
Using \eqref{eqlamb2} and \eqref{eqlambderiv}, $W(0)=0$ and $W'(0)=1$.
Now differentiating \eqref{eqlamb3}, we determine that for $x>-\frac{1}{e}$ $$W''(x)=-\frac{1+W'(x)e^{W(x)}}{(x+e^{W(x)})^{2}}$$ and $W''(0)=-2$. Finally, as $\lim_{x\to +\infty} W(x)=+\infty$, (\ref{eqlamb2}) provides that $$\lim_{x\to +\infty}\frac{\ln(x)}{W(x)}=1+\lim_{x\to +\infty}\frac{\ln(W(x))}{W(x)}=1.$$ $\;\square$
\end{proof}

\subsubsection{Proof of Theorem \ref{propdecompo}}
\label{proof_decompo_annexe}
Theorem \ref{propdecompo} provides a multiplicative decomposition of the Laplace transform of a function of a log-normal random variable. For that, 
we need to find the minimal value of $k_\beta$ (see \eqref{defkbeta}). This is the purpose of Lemma \ref{variation_k}. 
Fix $\theta, v>0.$ Let $m$ be defined on $[0,\infty]$ by 
$$
m(\beta) = \left\{
    \begin{array}{ll}
       -\frac{W(\theta v  \eta^2 T )}{\eta\sqrt{T}} & \mbox{if}\;\; \beta \in \left[ \frac1{\theta v} \left(\frac{W(\theta v  \eta^2 T)}{\eta^{2}T}
+ \frac{W^2(\theta v  \eta^2 T)}{2\eta^{2} T}\right),+\infty \right]  \\
        0 & \mbox{elsewhere}.
    \end{array}
\right.
$$

\begin{lemma}
\label{variation_k}
Let $\beta\in [0,+\infty]$ and $\theta, v>0$. Let $u$ be any real number if $\beta=+\infty$ and $u=-v\beta$ otherwise. Then, $k_\beta$ reaches its minimum value in $m(\beta)$. 
Moreover, 
\begin{eqnarray}
\label{min_k_beta}
k_\beta(m(\beta)) & = &-\left(\theta v \beta -\frac{W(\theta v  \eta^2 T)}{\eta^{2}T}
- \frac{W^2(\theta v  \eta^2 T)}{2\eta^{2} T}\right)_+\\
k_{+\infty}(m(+\infty)) & = & 
\frac{W(\theta v  \eta^2 T)}{\eta^{2}T}
+ \frac{W^2(\theta v  \eta^2 T)}{2\eta^{2}T}+\theta u.
\label{min_k_inf}
\end{eqnarray}
\end{lemma}

\begin{proof}[{Proof of Lemma \ref{variation_k}}]
We first study  $k : z\mapsto \theta \left(u+v e^{\eta\sqrt{T}z}\right)+\frac{z^2}{2}$. Then, we use \eqref{defkbeta} to link the minima of $k_\beta$ and  of $k$ according to the values of  $\beta$. 
As $v>0$, for $z\geq 0$, we have $k'(z)>0$. Letting $z<0$, as $W$ is strictly increasing and using \eqref{eqlamb2}, we obtain that 
\begin{eqnarray*}
k'(z) \ge 0 & \Leftrightarrow &  \theta v  \eta^2 T \ge - \eta\sqrt{T} z e^{-\eta\sqrt{T} z} 
 \Leftrightarrow  W(   \theta v  \eta^2 T) \ge W(- \eta\sqrt{T} z e^{-\eta\sqrt{T} z} )=-\eta\sqrt{T}z  \Leftrightarrow   z \geq - \frac{W(    \theta v  \eta^2 T)}{\eta\sqrt{T}},
\end{eqnarray*}
which gives the variations of $k$. Moreover, using \eqref{eqlamb2},
\begin{eqnarray*}
k \left(-\frac{W(\theta v  \eta^2 T )}{\eta\sqrt{T}} \right) &=&    \theta v e^{-W(\theta v  \eta^2 T) }
+ \frac{W^2(\theta v  \eta^2 T)}{2\eta^{2}T}+\theta u =  \frac{W(\theta v  \eta^2 T)}{\eta^{2}T}
+ \frac{W^2(\theta v  \eta^2 T)}{2\eta^{2} T}+\theta u.
\end{eqnarray*}
As $k_{\infty}=k$, we determine that $k_\beta$ reaches its minimum value in $m(+\infty)$ and that \eqref{min_k_inf} holds true.\\
We now study $k_\beta$ if $\beta<\infty$. As $u=-v\beta$, $k_\beta$ is continuous. 

Assume first that $\frac{W(\theta v  \eta^2 T)}{\eta^{2}T}
+ \frac{W^2(\theta v  \eta^2 T)}{2\eta^{2} T}\leq\theta v\beta$. 
We show that $k_\beta$ reaches its minimal value at $-\frac{W(\theta v \eta^2 T)}{\eta\sqrt{T}}.$ 
As $v,\, \theta>0,$ \eqref{eqlamb2} implies that 
\begin{eqnarray}
e^{-W(\theta v \eta^2 T)}=\frac{W(\theta v  \eta^2 T)}{\theta v \eta^2 T}\leq \beta.
\label{temp_min_1}
\end{eqnarray}
As \eqref{temp_min_1} is equivalent to $-\frac{W(\theta v \eta^2 T)}{\eta\sqrt{T}} \leq \frac{\ln \beta}{\eta\sqrt{T}}$, \eqref{defkbeta} and (\ref{min_k_inf}) show that
\begin{eqnarray}
\label{pulse}
k_\beta\left(-\frac{W(\theta v \eta^2 T)}{\eta\sqrt{T}}\right)&=&k\left(-\frac{W(\theta v \eta^2 T)}{\eta\sqrt{T}}\right)
 =  \frac{W(\theta v  \eta^2 T)}{\eta^{2}T}
+ \frac{W^2(\theta v  \eta^2 T)}{2\eta^{2} T}-\theta v \beta. 
\end{eqnarray} 
As $k_\beta=k$ on $(-\infty,\frac{\ln \beta }{\eta\sqrt{T}})$, it also shows that $k_\beta$ is decreasing on $(-\infty,-\frac{W(\theta v \eta^2 T)}{\eta\sqrt{T}})$ and increasing on $(-\frac{W(\theta v \eta^2 T)}{\eta\sqrt{T}},\frac{\ln \beta }{\eta\sqrt{T}})$.  
Now, on $(\frac{\ln \beta }{\eta\sqrt{T}},+\infty)$ $k_\beta(z)=\frac{z^2}2$, we distinguish two cases. First, assume that $\beta\geq 1$. Then, $-\frac{W(\theta v \eta^2 T)}{\eta\sqrt{T}} \leq 0 \leq \frac{\ln \beta }{\eta\sqrt{T}}$ and 
$k_\beta$ is increasing on $(\frac{\ln \beta }{\eta\sqrt{T}},+\infty)$. Thus, $k_\beta$ reaches its minimum value in $-\frac{W(\theta v \eta^2 T)}{\eta\sqrt{T}}.$  
If now $\beta<1$, using (\ref{temp_min_1}) again, we find that $-\frac{W(\theta v \eta^2 T)}{\eta\sqrt{T}} \leq\frac{\ln \beta }{\eta\sqrt{T}}\leq 0$. 
Therefore, $k_\beta$ is 
decreasing on $(\frac{\ln \beta }{\eta\sqrt{T}},0)$ and increasing on $(0,+\infty)$. Thus, $0$ and $-\frac{W(\theta v \eta^2 T)}{\eta\sqrt{T}}$ are two potential minima. However, using (\ref{pulse}), we determine that  
$ k_\beta(-\frac{W(\theta v \eta^2 T)}{\eta\sqrt{T}}) \leq  0 =k_\beta(0),$
and we conclude that $k_\beta$ reaches its minimal value at $-\frac{W(\theta v \eta^2 T)}{\eta\sqrt{T}}.$ \\
Suppose now that $\theta v\beta<\frac{W(\theta v  \eta^2 T)}{\eta^{2}T}
+ \frac{W^2(\theta v  \eta^2 T)}{2\eta^{2} T}$. This implies that $\beta < 1$. 
Indeed, using \eqref{eqlamb2} and $1+x\leq e^x$,
\begin{eqnarray*}
\beta < 
\frac{W(\theta v  \eta^2 T)}{\theta v\eta^{2}T}\left(1
+ \frac{W(\theta v  \eta^2 T)}{2}\right)
= e^{-W(\theta v \eta^2 T)}\left(1
+ \frac{W(\theta v  \eta^2 T)}{2}\right) \leq 1.
\end{eqnarray*}
We want to show that $k_\beta$ reaches its minimum value at $0.$ Then, as $\beta <1$, $k_\beta(0)=0.$ 

Assume first that $\beta\leq \frac{W(\theta v  \eta^2 T)}{\theta v\eta^{2}T}=e^{-W(\theta v \eta^2 T)}$, then $\frac{\ln \beta }{\eta\sqrt{T}}\leq -\frac{W(\theta v \eta^2 T)}{\eta\sqrt{T}}$. 
As $k_\beta=k$ on $(-\infty,\frac{\ln \beta }{\eta\sqrt{T}})$, $k_\beta$ is decreasing on $(-\infty,\frac{\ln \beta }{\eta\sqrt{T}})$. Moreover, on $(\frac{\ln \beta }{\eta\sqrt{T}},+\infty),$ $k_\beta(z)=\frac{z^2}2$ and $k_\beta$ is decreasing on $(\frac{\ln \beta }{\eta\sqrt{T}},0)$ and increasing on $(0,\infty).$
Thus, $k_\beta$ reaches its minimum value at $0$. \\
Assuming now that $\beta> \frac{W(\theta v  \eta^2 T)}{\theta v\eta^{2}T}=e^{-W(\theta v \eta^2 T)}$, we determine that $-\frac{W(\theta v \eta^2 T)}{\eta\sqrt{T}}<\frac{\ln \beta }{\eta\sqrt{T}}< 0$. Thus, as $k_\beta=k$ on $(-\infty,\frac{\ln \beta }{\eta\sqrt{T}}),$ $k_\beta$ is decreasing on $(-\infty,-\frac{W(\theta v \eta^2 T)}{\eta\sqrt{T}})$ and increasing on $(-\frac{W(\theta v \eta^2 T)}{\eta\sqrt{T}}, \frac{\ln \beta }{\eta\sqrt{T}}).$ Moreover, on $(\frac{\ln \beta }{\eta\sqrt{T}},\infty),$ $k_\beta(z)=\frac{z^2}2$, and $k_\beta$ is decreasing on $(\frac{\ln \beta }{\eta\sqrt{T}},0)$ and increasing on $(0,+\infty)$. Thus, $-\frac{W(\theta v \eta^2 T)}{\eta\sqrt{T}}$ and $0$ are two potential minima. Using \eqref{pulse}, we determine that $k_\beta(-\frac{W(\theta v \eta^2 T)}{\eta\sqrt{T}})>k_\beta(0)$
and we conclude that $k_\beta$ reaches its minimum in $0$. $\;\square$
\end{proof}

We are now in position to prove Theorem \ref{propdecompo}. 
\begin{proof}[{Proof of Theorem \ref{propdecompo}}]
We perform the Laplace method on (\ref{decompo_f}) and thus, get out the maximum of $e^{-k_\beta}$ of the integral. This gives the deterministic part of the decomposition. Then, an affine change of variable is applied to the remaining integral in order to retrieve the required expression.\\
Let $\beta\in [0,+\infty]$ and $\theta, v>0$. Let $u$ be any real number if $\beta=+\infty$ and $u=-v\beta$ otherwise. Assume that ($\ref{domaine}$) holds true. 
We set $\hat{\eta}=\eta\sqrt{T}$ and $\hat w=W(\theta v \hat{\eta}^2)$. 
Using (\ref{decompo_f}), we have that 
\begin{eqnarray}
L_{f,\beta}(\theta)=\exp(-k_\beta(m(\beta)))
\int_{\mathbb{R}}
\exp\left[
-\theta(f(e^{\hat{\eta} z})-u -v e^{\hat{\eta} z})1_{z\leq \frac{\ln(\beta)}{\hat{\eta}}} -  (k_{\beta}(z)-k_\beta(m(\beta)))
\right]
\frac{dz}{\sqrt{2\pi}}.
\label{integ}
\end{eqnarray}
Recall that when $\beta<+\infty$, we have chosen $u=-v\beta$. Thus, (\ref{domaine}) and Lemma \ref{variation_k} imply that $m(\beta)=-\hat w/\hat{\eta}.$ Using (\ref{min_k_beta}) and (\ref{min_k_inf}), we determine that 
\begin{eqnarray}
\exp\left[-k_\beta(m(\beta))\right] & = & \exp{\left[-\left(\theta u + \frac{\hat w}{\hat{\eta}^2}+\frac{\hat w^2}{2\hat{\eta}^2}\right)\right]}={L}_\beta(\theta).
\label{equality_ek}
\end{eqnarray}
Then, (\ref{integ}) and (\ref{equality_ek}) with the change of variable $z=y-m(\beta)=y+\hat w/\hat{\eta}$ and the fact that $e^{-\hat w}=\hat  w/(\theta v \hat{\eta}^2)$ (see (\ref{eqlamb2})) 
imply that 
\begin{align}
L_{f,\beta}(\theta) 
&= \begin{multlined}[t][\textwidth-2.5cm]{L}_\beta(\theta)\int_{\mathbb{R}}\exp\bigg[-\theta\left(f\left(\frac{\hat w}{\theta v \hat{\eta}^2}e^{\hat{\eta} y}\right)-u -v \frac{\hat w}{\theta v\hat{\eta}^2}e^{\hat{\eta} y}\right)1_{y\leq \frac{\ln \beta}{\hat{\eta}}+\frac{\hat w}{\hat{\eta}}}-\left(k_{\beta}\left(y-\frac{\hat w}{\hat{\eta}}\right)-k_\beta\left(-\frac{\hat w}{\hat{\eta}}\right)\right)\bigg]\frac{dz}{\sqrt{2\pi}}. \end{multlined}
\label{integ2}
\end{align}
Letting $y\in \mathbb{R}$, we find that
\begin{eqnarray*}
k_\beta\left(y-\frac{\hat w}{\hat{\eta}}\right) 
&=& \theta(u+ ve^{\hat{\eta}y}e^{-\hat w})1_{y\leq \frac{\ln \beta}{\hat{\eta}}+\frac{\hat w}{\hat{\eta}}}+\frac{(y-\frac{\hat w}{\hat{\eta}})^2}{2}= \theta u 1_{y\leq \frac{\ln \beta}{\hat{\eta}}+\frac{\hat w}{\hat{\eta}}} +\frac{\hat w}{\hat{\eta}^2}e^{\hat{\eta}y}1_{y\leq \frac{\ln \beta}{\hat{\eta}}+\frac{\hat w}{\hat{\eta}}}+\frac{\hat w^2}{2\hat{\eta}^2}-\frac{\hat w}{\hat{\eta}}y+\frac{y^2}{2}.
\end{eqnarray*}
We remark that $\frac{\ln \beta}{\hat{\eta}}+\frac{\hat w}{\hat{\eta}} \geq 0$. This is trivially true if $\beta=+\infty$. If $\beta<+\infty,$ as $u=-v\beta$ 
\begin{eqnarray*}
\frac{\ln \beta}{\hat{\eta}}+ \frac{\hat w}{\hat{\eta}} \geq 0 & \Leftrightarrow & \beta e^{\hat w}= \beta \frac{\theta v  \hat{\eta}^2}{\hat w}  \geq 1
 \Leftrightarrow  \theta v \beta \geq \frac{\hat  w}{\hat{\eta}^2},
\end{eqnarray*}
which is true by (\ref{domaine}). Thus, $k_\beta(-\frac{\hat w}{\hat{\eta}})   =  \theta u  + \frac{\hat w}{\hat{\eta}^2}+ \frac{\hat w^2}{2\hat{\eta}^2}$ and
\begin{eqnarray*}
k_\beta\left(y-\frac{\hat w}{\hat{\eta}}\right)-k_\beta\left(-\frac{\hat w}{\hat{\eta}}\right) 
&=& \frac{\hat w}{\hat{\eta}^2}(e^{\hat{\eta} y}-1-\hat{\eta} y) -\left(\theta u+\frac{\hat w}{\hat{\eta}^2}e^{\hat{\eta} y}\right) 1_{y>\frac{\ln \beta}{\hat{\eta}}+\frac{\hat w}{\hat{\eta}}}+\frac{y^2}{2}.\nonumber
\end{eqnarray*}
Assume that $\beta<+\infty.$ Then, 
\begin{eqnarray*}
y>\frac{\ln \beta}{\hat{\eta}}+ \frac{\hat w}{\hat{\eta}} & \Leftrightarrow & e^{\hat{\eta}y}> \beta e^{\hat w}= \beta \frac{\theta v  \hat{\eta}^2}{\hat w} 
 \Leftrightarrow  \frac{\hat  w}{\hat{\eta}^2}e^{\hat{\eta}y}> \theta v \beta. 
\end{eqnarray*}
Thus, as $u=-v\beta$,
$$\left(\theta u+\frac{\hat w}{\hat{\eta}^2}e^{\hat{\eta}y}\right)1_{y>\frac{\ln \beta}{\hat{\eta}}+\frac{\hat w}{\hat{\eta}}}
= \left(\frac{\hat w}{\hat{\eta}^2}e^{\hat{\eta}y}-\theta v \beta\right)1_{\frac{\hat w}{\hat{\eta}^2}e^{\hat{\eta}y}- \theta v \beta>0}
=\left(\frac{\hat w}{\hat{\eta}^2}e^{\hat{\eta}y}-\theta v \beta\right)_+.$$
The same obviously holds true when $\beta=+\infty,$ with the convention $(-\infty)_+=0$. Thus, we get that 
\begin{eqnarray*}
\exp \left[-\left(k_\beta \left(y-\frac{\hat w}{\hat{\eta}}\right)-k_\beta\left(-\frac{\hat w}{\hat{\eta}}\right)\right) \right]&=& \phi_{\beta}(y,\theta)e^{-\frac{y^2}{2}}.
\end{eqnarray*}
Therefore, (\ref{integ2}) implies that 
\begin{eqnarray*}
L_{f,\beta}(\theta)&=&{L}_\beta(\theta) \int_{\mathbb{R}}\exp\left[-\theta\left(f\left(\frac{\hat w}{\theta v \hat{\eta}^2}e^{\hat{\eta}y}\right)-u-\frac{\hat w}{\theta  \hat{\eta}^2}e^{\hat{\eta} y}\right)1_{y\leq \frac{\ln \beta}{\hat{\eta}}+\frac{\hat w}{\hat{\eta}}}\right] \phi_\beta(y,\theta)e^{-\frac{y^2}{2}} \frac{dy}{\sqrt{2\pi}}={L}_\beta(\theta)I_{f,\beta}(\theta). \quad\square
\end{eqnarray*}
\end{proof}

\subsection{Proofs of Section \ref{seclongstock}}
\label{proof_2}
For ease of notation, we introduce the functions $\theta : \rho\in (-1,1)\mapsto \lambda\gamma(1-\rho^2),$  
$\widetilde{\theta} : \gamma \in (0,+\infty)\mapsto \lambda\gamma(1-\rho^2)$ 
and 
\begin{align}
w  \;:\;&  \rho\in (-1,1)\mapsto W\left(s_0 \eta^2 Te^{\left(\nu-\eta\rho s_R -\frac{\eta^2}{2}\right)T}  \theta (\r)\right) \label{eqw} \\
\widetilde w  \;:\; & \gamma\in (0,+\infty)\mapsto W\left(s_0 \eta^2 Te^{\left(\nu-\eta\rho s_R -\frac{\eta^2}{2}\right)T}  \widetilde{\theta} (\gamma)\right).
 \end{align} 
 Then, $A=a(\rho)$, $B=b(\rho)$, $D=d(\rho)$, $G=g(\rho)$ and $D=\widetilde{d}(\gamma)$, where
 \begin{align}
 d\; :\;& \rho\in (-1,1) \mapsto 
 \frac{  e^{-rT}}{\gamma\eta^2 T(1-\rho^2)} \left( {w}(\rho)+\frac{w^2 (\r)}{2}\right) \label{eqpetitd} \\
 a\; :\;&\rho\in\ (-1,1)\mapsto -\frac{ e^{-rT}}{\gamma(1-\rho^2)}\ln  \mathbb{E} \left(\exp\left[- \frac{{w}(\rho)}{\eta^2 T}\left(e^{\eta\sqrt{T}N}-1-\eta\sqrt{T}N\right)\right]\right)
\nonumber\\
 b\; :\;&\rho\in\ (-1,1)\mapsto\frac{e^{-rT}}{\gamma(1-\rho^2)}\frac{{w}(\rho)}{\eta^2 T} \left(e^{\frac{\eta^2}2T}-1\right)\label{B}\\
 g\; :\;& \rho\in (-1,1) \mapsto d(\rho)+b(\rho)= \frac{ e^{-rT}}{\gamma(1-\rho^2)}\frac{w(\rho)}{\eta^2 T}\left(e^{\frac{\eta^2}{2} T}+\frac{w(\rho)}{2}\right) \label{g}\\
  \widetilde d\; :\;& \gamma\in (0,+\infty) \mapsto 
 \frac{ \lambda e^{-rT}}{\widetilde{\theta}(\gamma)\eta^2 T} \left( \widetilde w (\g)+\frac{\widetilde w^2 (\g)}{2}\right)
 \end{align} 

\subsubsection{Proofs of Subsection \ref{subappp}}
\label{sub_52}

We start with the proof of Theorem \ref{Evry}. Proposition \ref{propa} gives some asymptotics of $d$ and $g$, while Proposition \ref{sensi} gives some greeks related to $d$.

\begin{proof}[Proof of Theorem \ref{Evry}]
Assume for a moment that we have proved that  for all $\rho\in(-1,1),$ 
\begin{eqnarray} \label{eqab}\left(b(\r)-\frac{ e^{-r T}}{\gamma(1-\rho^2)}\frac{{w}(\rho)^2}{2\eta^4 T^2}e_2\right)_+\le a(\r) \le b(\r)\end{eqnarray} 
Then, as $p=D+A$, Remark \ref{lower_bound_D}, \eqref{g} and \eqref{eqab} show that $D \leq p \leq d(\rho) +b(\rho)=G.$
Then, $V_D(x_0,\l)\leq V(x_0,\l) \leq V_G(x_0,\l)$ follows from \eqref{hat_D_value} and \eqref{value_pe}.  

We prove now \eqref{eqab}. Using the Jensen inequality for the convex function $x \mapsto e^{-x}$, we obtain that 
\begin{eqnarray*}
 \mathbb{E}\left(\exp \left[-\frac{{w}(\rho)}{\eta ^2 T} \left(e^{\eta \sqrt T N}-1-\eta \sqrt T N\right) 
\right]\right) & \geq &  
  \exp \left[-\frac{{w}(\rho)}{\eta ^2 T}\left(e^{\frac{\eta^2}2T}-1\right)
\right].
\end{eqnarray*}
And so, using Remark \ref{lower_bound_D}, we determine that 
 $0 \leq a(\r) \leq  \frac{ \lambda e^{-rT}}{\theta(\rho)}\frac{{w}(\rho)}{\eta^2 T} \left(e^{\frac{\eta^2}2T}-1\right)=b(\r).$
Then, as $e^{-x}\leq 1-x+\frac{x^2}{2}$ for all $x\geq 0$, we have that
\begin{eqnarray*}
 \mathbb{E}\left(\exp \left[-\frac{{w}(\rho)}{\eta ^2 T} \left(e^{\eta \sqrt T N}-1-\eta \sqrt T N\right) 
\right]\right)\leq 1-\frac{{w}(\rho)}{\eta^2 T}\left(e^{\frac{\eta^2 T}{2}}-1\right)+\frac{{w}(\rho)^2}{2\eta^4 T^2}e_2,
\end{eqnarray*} 
where 
$e_2=\mathbb{E}\left(\left(e^{\eta\sqrt{T}N}-1-\eta\sqrt{T}N\right)^2\right) = e^{2\eta^2 T}-2(1+\eta^2  T)e^{\frac{\eta^2 T}{2}}+\eta^2 T +1.$
Composing with the logarithm function and using the inequality $\ln(1+x)\leq x$ yields
\begin{eqnarray}
 a(\r) \ge b(\r)-\frac{\lambda e^{-r T}}{\theta(\rho)}\frac{{w}(\rho)^2}{2\eta^4 T^2}e_2,
\label{ineq_A}
\end{eqnarray}
which combined with $a(\r)\geq 0$ yields in turn the lower bound for $a(\r)$.\\
Now, we prove \eqref{ineq1} and \eqref{ineq2}. 
As $ D\leq p=D+A \le D+B,$ we find that 
\begin{eqnarray*}
1\leq \frac{p}{D} \le 1+ \frac{e^{\frac{\eta^2T}{2}}-1}{1+\frac{{w}(\rho)}{2}}&\iff &  \frac{1+\frac{{w}(\r)}{2}}{e^{\frac{\eta^2T}{2}}+\frac{{w}(\r)}{2}}\le\frac{D}{p}\le 1,\\
\frac{1+\frac{{w}(\r)}{2}}{e^{\frac{\eta^2T}{2}}+\frac{{w}(\r)}{2}} & = & e^{-\frac{\eta^2T}{2}}\frac{e^{\frac{\eta^2T}{2}}+\frac{{w}(\r)}{2}e^{\frac{\eta^2T}{2}}}{e^{\frac{\eta^2T}{2}}+\frac{{w}(\r)}{2}}\geq   e^{-\frac{\eta^2T}{2}},
\end{eqnarray*}
as ${w}(\r)\geq 0$. Moreover,  using \eqref{g} and (\ref{ineq_A}), 
$p=G+A-B \geq G-\frac{\lambda e^{-r T}}{\theta(\rho)}\frac{{w}(\rho)^2}{2\eta^4 T^2}e_2.$
As $p\geq D$, using (\ref{eqpetitd}), we obtain that 
\begin{eqnarray*}
\frac{G}{p} \leq 1+ \frac{\lambda e^{-r T}}{\theta(\rho)}\frac{{w}(\rho)^2}{2\eta^4 T^2 p}e_2 \leq 1+ \frac{\lambda e^{-r T}}{\theta(\rho)}\frac{{w}(\rho)^2}{2\eta^4 T^2 D}e_2 = 1+\frac{{w}(\rho)e_2}{2\eta^2 T+\eta^2 T{w}(\rho)}\leq 1+\frac{e_2}{\eta^2 T}. \quad\square
\end{eqnarray*}
\end{proof}

\begin{proof}[{Proof of Proposition \ref{propa}}] 
The proof relies on the properties of the Lambert function given in Lemma \ref{lemlamb}. 
As ${w}(\r) > 0$, we find that $d>0$ and $g>0$. Lemma \ref{lemlamb} imply that $d$ and $g$ are $C^1$ and that 
\begin{eqnarray}
\label{lim}
\frac{w (\r)}{\gamma (1-\rho^2) \eta^2T} & \underset{\rho\rightarrow 1^-}{\sim} & 
  s_0 e^{\left(\nu-\eta s_R -\frac{\eta^2}{2}\right)T}. 
\end{eqnarray}
As $\lim_{\rho\mapsto 1^-}{w}( \r)=0$, the first limits in \eqref{eqd0} and \eqref{eqg0} are proved. The case when $\rho$ goes to $-1^+$ is treated similarly. We deduce that $d$ and $g$ are bounded on $(-1,1)$.

Recalling (\ref{lim}), we obtain that
\begin{eqnarray*}
\lim_{\rho\rightarrow 1^-}b(\rho)& = &\lambda e^{-rT} s_0 e^{\left(\nu-\eta s_R-\frac{\eta^2}{2}\right)T}\left(e^{\frac{\eta^2T}{2}}-1\right) \quad \lim_{\rho\rightarrow -1^+}b(\rho)  =  \lambda e^{-rT} s_0 e^{\left(\nu+\eta s_R-\frac{\eta^2}{2}\right)T}\left(e^{\frac{\eta^2T}{2}}-1\right).
\end{eqnarray*}
Since $b$ is continuous on $(-1,1)$ (recall that $W$ is continuous), $b$ is bounded, and since $0\leq a \leq b$, $a$ is also bounded.

Using Lemma \ref{lemlamb}, we determine that $\widetilde w(\gamma) \underset{\gamma\to 0^+}{\sim}  s_0 \eta^2 T e^{(\nu-\eta\rho s_R-\frac{\eta^2}{2})T}\widetilde{\theta}(\g)$ and 
$\widetilde w(\gamma) \underset{\gamma\to +\infty}{\sim}  \ln(\gamma).$ Thus,
\begin{eqnarray*}
\widetilde d(\gamma) &\underset{\gamma\to 0^+}{\sim} & \l e^{-rT}  s_0 e^{(\nu-\eta\rho s_R-\frac{\eta^2}{2})T}\left(1+\frac{\widetilde{w}(\gamma)}{2}\right) 
  \underset{\gamma\to 0^+}{\sim}  \l e^{-rT} s_0 e^{(\nu-\eta\rho s_R-\frac{\eta^2}{2})T}  \underset{\gamma\to +\infty}{\sim} \frac{ \lambda e^{-rT}\ln^2( \g)}{2\widetilde{\theta}(\gamma)\eta^2 T} \underset{\gamma\to +\infty}{\to} 0. \quad\square
\end{eqnarray*}

\end{proof}

\begin{proof}[Proof of Proposition  \ref{sensi}]
The proof consists in the differentiation of $d$ and $\widetilde{d}$ together with Lemma \ref{lemlamb} in order to simplify the equalities involving the Lambert function. 
Let $\delta : \rho\in[-1,1]\mapsto \nu-\eta\rho s_R$. We easily determine that for all $\rho\in (-1,1)$, 
 \begin{eqnarray*}
d'(\rho) & = & 
\frac{ \l e^{-rT} }{\eta^2 T \theta(\r)}\left(   w' (\r)(1+ {w (\r)})  
- \frac{\theta'(\r)}{\theta(\r)} w (\r)\left( 1+ \frac{w (\r)}2\right)\right). 
\end{eqnarray*}
As $s_0 \eta^2 Te^{(\delta(\rho) -\frac{\eta^2}{2})T}  \theta (\r)  \geq 0$,  \eqref{eqlambderiv} and \eqref{eqw} imply that 
\begin{align*}
w'(\r) & =  s_0 \eta^2 Te^{\left(\delta(\rho) -\frac{\eta^2}{2}\right)T}  \theta (\r) W'\left(s_0 \eta^2 Te^{\left(\delta(\rho) -\frac{\eta^2}{2}\right)T}  \theta (\r) \right) \left(\frac{\theta'(\r)}{\theta(\r)} - \eta s_R T\right)\\
& =  \frac{\frac{\theta'(\r)}{\theta(\r)} -\eta s_R T }{1+ \frac1{W\left( s_0 \eta^2 Te^{(\delta(\rho) -\frac{\eta^2}{2})T}  \theta (\r) \right)}}=\frac{\frac{\theta'(\r)}{\theta(\r)} -\eta s_R T }{1+ \frac1{w(\r)}}. 
\end{align*}
It follows that
\begin{eqnarray}
\label{sign_delta_rho}
\Delta_\rho = d'(\rho) & = & 
\frac{ \l e^{-rT} w(\r)}{\eta^2 T \theta(\r)}\left(  \frac{\theta'(\r)}{\theta(\r)} -\eta s_R T 
- \frac{\theta'(\r)}{\theta(\r)} \left( 1+ \frac{w (\r)}2\right)\right) = -\frac{ \l e^{-rT} w(\r)}{\eta^2 T \theta(\r)}\left(  \eta s_R T - \frac{\rho}{1-\rho^2}{w}(\rho)\right) .
\end{eqnarray}
Using (\ref{eqlambderiv}), we see that
$\widetilde w'(\gamma)=\frac{\widetilde{\theta}'(\g)}{\widetilde{\theta}(\g)}\frac{\widetilde w(\gamma)}{1+\widetilde w(\gamma)}.$
The function $\widetilde d$ is $C^{1}$ (see Lemma \ref{lemlamb}) and we easily ascertain that
\begin{eqnarray*}
\Delta_\gamma=\widetilde{d}'(\g) &=& \frac{\l e^{-rT}}{\eta^2 T}\left(\frac{\widetilde w'(\g)}{\widetilde{\theta}(\g)}(1+\widetilde w(\g))-\frac{\widetilde{\theta}'(\g)}{\widetilde{\theta}^2(\g)}\left(\widetilde w(\g)+\frac{\widetilde w^2(\g)}{2}\right) \right) = -\frac{\l e^{-rT}}{2\widetilde \theta^2(\gamma)\eta^2 T}\widetilde w^2(\gamma)\widetilde{\theta}'(\gamma)\\
&=&  -\frac{ e^{-rT}}{2 \eta^2 T\gamma^2(1-\rho^2)}\widetilde w^2(\gamma). \quad\square
\end{eqnarray*}
\end{proof}

\subsubsection{Proofs of Subsection \ref{subvard}}
\label{proof_3}
 Proposition  \ref{variation} gives the expression of the minimum of $d$ in $\rho$ as well as the value of $d$ at that minimum. The resulting minimum is then expanded in Lemma \ref{closed} using the Taylor expansion of the Lambert function given in Lemma \ref{lemlamb}.
Proposition \ref{rho_star_strat} highlights the link between the deterministic strategy given in (\ref{new_strat}) and $\rho^*$. Proposition \ref{short_mat_strat} gives the expression of the correlation between $S$ and $P$ for which an agent would not invest much in $P$ in the case of a short maturity.  \newpage

\begin{proof}[Proof of Proposition \ref{variation}] 

As ${w}(\rho)> 0$ and $\theta(\rho)> 0$ for all $\rho \in (-1,1)$, (\ref{sign_delta_rho}) shows that the sign of $d'(\rho)$ is the sign of $\frac{\rho}{1-\rho^2} {w}(\rho)-\eta s_R T$.  Supposing first that $\mu\geq r$, then for $\rho\in(-1,0],$ $d'(\rho)\leq 0$.
If now $\rho\in (0,1)$
\begin{eqnarray*}
d'(\rho)\geq  0 & \iff & W\left(s_0\eta^2 Te^{\left(\delta(\rho) -\frac{\eta^2}{2}\right)T}  \theta (\r) \right) \geq  
 \eta  s_R T \frac{1-\rho^2}{\r}\\
  & \iff & s_0 \eta^2 Te^{\left(\nu-\eta\rho s_R-\frac{\eta^2}{2}\right)T}  \theta (\r)  \geq \eta  s_R T \frac{1-\rho^2}{\r} e^{\eta s_R \frac{1-\rho^2}{\r}T}\\
   & \iff &  \l \g    s_0 \eta^2  Te^{\left(\n  -\frac{\eta^2}{2}\right)T} \geq \eta \frac{s_R}{\r}T  e^{\eta \frac{s_R}{ \r}T }\\
    & \iff & W\left( \l  \g s_0 \eta^2   Te^{\left(\n  -\frac{\eta^2}{2}\right)T}  \right)\geq \eta \frac{s_R}{\r}T \\
& \iff & \rho\ge\rho^*. 
\end{eqnarray*}
For the second equivalence, we have composed  by the reciprocal function of $W$ 
and for the fourth one by $W$ (see \eqref{eqlamb2}). 

Suppose now that $\mu\leq r$. Then, for $\rho\in[0,1),$ $d'(\rho)\geq 0$.
If now $\rho\in (-1,0)$, we prove as before that $d'(\rho)\geq 0$ if and only if $\rho\ge\rho^*$.  
The results for the variation of $d$ are as follow. 

Setting $\k= \l \gamma s_0 \eta^2 Te^{(\nu -\frac{\eta^2}{2})T}>0,$ 
then $\r^*=(\eta s_R T)/W(\k).$ Assume now that $-1<\rho^*<1$,
\begin{eqnarray}
w(\r^*)
& = & W\left(\k e^{-\eta\rho^* s_R T}(1-\r^{* 2})\right)
=  W\left(e^{-\frac{\left(\eta s_R T\right)^2}{W(\k)}} \k\frac{W^2(\k)-\left(\eta s_R T\right)^2}{W^2(\k)}\right) \nonumber\\
& = &W\left(e^{-\frac{\left(\eta s_R T\right)^2}{W(\k)} + W(\k)} \left(-\frac{\left(\eta s_R T\right)^2}{W(\k)} + W(\k)\right)\right) = -\frac{\left(\eta s_R T\right)^2}{W(\k)} + W(\k)\nonumber\\ &=& W(\k)(1-\r^{* 2})= {w}(0)(1-\rho^{* 2})\label{wstar}, 
\end{eqnarray}
where we have used \eqref{eqlamb2} for the third equality (as $\k>0$) and for the fourth one as $-1<\rho^*<1$, and thus 
$-\frac{\left(\eta  s_R T\right)^2}{W(\k)} + W(\k)>0$. 
It follows that 
\begin{eqnarray*}
d(\rho^*)& = & \frac{e^{-rT}}{\gamma \eta^2 T (1-\rho^{* 2})} w (\r^*)\left(1 +\frac{w (\r^*)}{2}\right)
 = \frac{ e^{-rT}}{ \g \eta^2 T } W(\k)\left(1 +\frac{1}{2}W(\k)\left(1-\frac{\left(\eta s_R T\right)^2}{W^2(\k)}\right)\right) \\
& = & \frac{ e^{-rT}}{ \g \eta^2 T } \left(W(\k) +\frac{W^2(\k)}{2} -\frac{1}{2}\left(\eta s_R T\right)^2\right) 
 =  d(0)-\frac{e^{-rT}}{2\gamma} s_R^2 T.\end{eqnarray*}

Using \eqref{B}, \eqref{g} and (\ref{wstar}), we obtain that 
\begin{align*}
b(\rho^*) &=  \frac{e^{-rT}}{\gamma (1-\rho^{* 2})}\frac{{w}(\rho^*)}{\eta^2 T} \left(e^{\frac{\eta^2}2T}-1\right)=\frac{ e^{-rT}}{\gamma}\frac{{w}(0)}{\eta^2 T} \left(e^{\frac{\eta^2}2T}-1\right)=b(0)\\
g(\r^*)  &= d(\r^*)+b(\r^*) = d(0)+b(0)-\frac{e^{-rT}}{2\gamma} s_R^2 T = g(0)-\frac{e^{-rT}}{2\gamma} s_R^2 T.
\end{align*}

So, using Remark \ref{lower_bound_D}, we determine that 
\begin{eqnarray*}
p=d(\r)+a(\r) \geq  d(\r)\geq 
\left\{
    \begin{array}{ll}
        d(\rho^*)= d(0)-\frac{e^{-rT}}{2\g}s_R^2 T & \mbox{if } -1<\rho^*<1 \\
        \lim_{\rho\to 1^-}d(\rho) & \mbox{if}\; \rho^* \geq 1\\
        \lim_{\rho\to -1^+}d(\rho) & \mbox{if}\; \rho^* \leq -1.
    \end{array}
\right.
\end{eqnarray*}
and \eqref{unifp} follows from (\ref{eqd0}). Then, \eqref{value_pe} implies \eqref{unifv}. $\;\square$
\end{proof}

%
%
%

\begin{proof}[{Proof of Proposition \ref{rho_star_strat}}]
Assume that $-1<\rho^*<1$. Using (\ref{new_strat}) and (\ref{wstar}), we get that
\begin{eqnarray*}
{\Pi}^{D,\lambda}_{\rho^*}(0,s_0)&=& e^{-r T}\left(\frac{s_R}{\gamma\sigma}-\frac{\rho^*\; w(\rho^*)}{\sigma\eta\g (1-\rho^{* 2})T}\right)= e^{-r T}\left(\frac{s_R}{\gamma\sigma}-\frac{\rho^*\; w(0)}{\sigma\eta\g T}\right)= e^{-rT}\left(\frac{s_R}{\gamma\sigma}-\frac{\eta s_R T\;}{\sigma\eta\g T}\right)= 0. \quad\square
\end{eqnarray*}
\end{proof} 

\begin{proof}[{Proof of Proposition \ref{short_mat_strat}}]
Using (\ref{sup_strat}), we have to prove that $\lim_{T\to 0^+}\frac{\partial p}{\partial s}(\rho,s_0)=\lambda$. Differentiating in (\ref{indiffdef_gen}) with $h=id$, we get that
\begin{eqnarray*}
\frac{\partial p}{\partial s}(\rho,s_0) &=& \lambda e^{-rT} \frac{\mathbb{E}\left(e^{\left(\nu-\eta\rho s_R-\frac{\eta^2}{2}\right)T+\eta\sqrt{T}N} \exp\left[-\l\g(1-\r^2)s_0 e^{\left(\nu-\eta\rho s_R-\frac{\eta^2}{2}\right)T+\eta\sqrt{T}N}\right]\right)}{\mathbb{E}\left(\exp\left[-\l\g(1-\r^2)s_0 e^{\left(\nu-\eta\rho s_R-\frac{\eta^2}{2}\right)T+\eta\sqrt{T}N}\right]\right)}\underset{T\to 0^+}{\longrightarrow} \lambda,
\end{eqnarray*}
where we have used the dominated convergence theorem. The second assertion is trivial. $\;\square$
\end{proof}

\begin{proof}[{Proof of Lemma \ref{closed}}] 
We use the notation $\epsilon$ for any function vanishing when $T\to 0^+$. Using Lemma \ref{lemlamb}, we obtain that 
\begin{align*}
W\left(\l\g\eta^2 T s_0 e^{\left(\nu-\frac{\eta^2}{2}\right)T}\right)&= \l\g\eta^2 T s_0 e^{(\nu-\frac{\eta^2}{2})T} \left(1-\l\g\eta^2 T s_0 e^{\left(\nu-\frac{\eta^2}{2}\right)T}+T\epsilon(T)\right)\\
&= \l\g\eta^2 T s_0 \left(1+\left(\nu-\frac{\eta^2}{2}\right)T\right) \left(1-\l\g\eta^2 T s_0\right)+T\epsilon(T)\\
&= \l\g\eta^2 T s_0 \left(1-\left(\l \g\eta^2 s_0 -\left(\nu-\frac{\eta^2}{2}\right)\right)T+T\epsilon(T)\right).
\end{align*}
Using \eqref{minimum}, we thus obtain that 
\begin{eqnarray*}
\rho^* 
&=& s_R\frac{1}{\l\g\eta s_0}\left(1+\left(\l \g\eta^2 s_0 -\left(\nu-\frac{\eta^2}{2}\right)\right)T+T\epsilon(T)\right).
\end{eqnarray*}
Now, as $\nu> \frac{\eta^2}{2},$ Lemma \ref{lemlamb} implies again that 
\begin{eqnarray*}
W\left(\l\g\eta^2 T s_0 e^{\left(\nu-\frac{\eta^2}{2}\right)T}\right) \underset{T\to +\infty}{\sim}&\left(\nu-\frac{\eta^2}{2}\right)T \quad \mbox{and} \quad 
\rho^* \underset{T\to +\infty}{\sim}&s_R \frac{\eta}{\nu-\frac{\eta^2}{2}}. \quad \square
\end{eqnarray*}
\end{proof}

\section{Declarations}
\subsection{Author’s Contribution}
All authors contributed equally to this document. The numerical analysis was performed by the corresponding author. All authors read and approved the final manuscript.

\subsection{Competing interests}
The authors have no competing interests or other interests that might be perceived to influence the results and the discussion reported in this paper.

\subsection{Availability of Data and Materials}
The datasets analysed during the current study are available at the web address : finance.yahoo.com

\subsection{Funding}
The authors thank De Vinci Research Center  for its funding.

\end{document}